\documentclass{amsart}
\usepackage{amsthm}
\usepackage{amssymb}
\usepackage{amsaddr}
\usepackage{thmtools}
\usepackage[capitalise]{cleveref}
\usepackage{relsize}

\AtBeginDocument{%
  }

\newcommand{\set}[1]{{\left\{ {#1} \right\} }}
\newcommand{\m}[1]{{\uppercase {\bf{#1}}}}
\newcommand{\con}[1]{{\sf Con\:\m{#1}}}

\newcommand{\tn} [1]{{\mathbf{#1}}}
\newcommand{\typ}{{\rm typ}}
\newcommand{\typset}[1]{\typ\set{#1}}
\newcommand{\pol}[1]{{\rm Pol\:\m #1}}
\newcommand{\poln}[2]{{\rm Pol_{#1}\:\m #2}}
\newcommand{\po}[1]{{\mathbf{#1}}}
\newcommand{\ar}[1]{\text{ar(}#1\text{)}}

\newtheorem{theorem}{Theorem}
\newtheorem{lemma}[theorem]{Lemma}
\newtheorem{definition}[theorem]{Definition}
\newtheorem{corollary}[theorem]{Corollary}
\newtheorem{observation}[theorem]{Observation}
\newtheorem{fact}[theorem]{Fact}
\crefname{fact}{Fact}{Facts}
\newtheorem{open}{Open Problem}
\newtheorem{example}[theorem]{Example}

\newcommand{\ptime}{\textsf{P}}
\newcommand{\nptime}{\mathsf{NP}}

\newcommand{\ov}[1]{\overline{#1}}

\newcommand{\ppoly}{\mathsf{P}/poly}
\newcommand{\palg}[1]{nu\mathsf{P}_{#1}}

\newcommand{\progb}[3]{\left(#1\right)\!\left[#2,#3\right]}
\newcommand{\prog}[2]{\left(#1\right)\!\left[#2\right]}

\newcommand{\MODG}[1]{\textsc{MOD}_{#1}}
\newcommand{\ANDG}[1]{\textsc{AND}_{#1}}
\newcommand{\ORG}[1]{\textsc{OR}_{#1}}
\newcommand{\PROGG}[1]{\textsc{PROG}_{#1}}
\newcommand{\BABPG}[1]{\textsc{GABP}_{#1}}
\newcommand{\BOUNDG}[1]{\textsc{BAR}_{#1}}
\newcommand{\circuittwo}[2]{#1 \circ #2}
\newcommand{\circuit}[3]{#1 \circ #2 \circ #3}
\newcommand{\circuitfour}[4]{#1 \circ #2 \circ #3 \circ #4}
\newcommand{\BOUNDA}[1]{\textsc{BAR}_{#1}}
\newcommand{\MON}{\textsc{MONOTONE}}
\newcommand{\MMON}{\textsc{MULTIMONOTONE}}
\newcommand*{\defeq}{\stackrel{\mathsmaller{\mathsf{def}}}{=}}
\newcommand{\nudet}[1]{\mathsf{nuDet}_{#1}}
\newcommand{\nuitdet}{\mathsf{nuItDet}}

\RequirePackage[
  datamodel=acmdatamodel,
  style=acmnumeric,
  ]{biblatex}

\begin{document}

\title{Complexity Classes Arising from Circuits over Finite Algebraic Structures}

\author{Piotr Kawałek}
\email{piotr.kawalek@tuwien.ac.at}
\address{Institute of Discrete Mathematics and Geometry,\\ TU Wien,\\ Vienna, Austria}
\thanks{Piotr Kawałek: This research was funded in whole or in part by National Science Centre, Poland \#2021/41/N/ST6/03907. For the purpose of Open Access, the author has applied a CC-BY public copyright licence to
 any Author Accepted Manuscript (AAM) version arising from this submission.
Funded by the European Union (ERC, POCOCOP, 101071674). Views
 and opinions expressed are however those of the author(s) only and do not necessarily reflect those
 of the European Union or the European Research Council Executive Agency. Neither the European
 Union nor the granting authority can be held responsible for them.}
%\orcid{0000-0003-3592-1697}
%\affiliation{%
 % \institution{Institute of Discrete Mathematics and Geometry, TU Wien}
 % \city{Vienna}
 % \country{Austria}
%}

\author{Jacek Krzaczkowski}
\email{krzacz@umcs.pl}
\address{Department of Computer Science and Mathematics,\\ Maria Curie-Sk\l{}odowska University,\\ Lublin, Poland}
\thanks{Jacek Krzaczkowski: This research was funded in whole or in part by National Science Centre, Poland \#2022/45/B/ST6/02229. For the purpose of Open Access, the author has applied a CC-BY public copyright licence to any Author Accepted Manuscript (AAM) version arising from this submission.}

%\orcidlink{0000-0003-2861-1156}
%\affiliation{\institution{Department of Computer Science and Mathematics,\\ Maria Curie-Sk\l{}odowska University}
%\city{Lublin}
%\country{Poland}}

\begin{abstract}
Most classical results in circuit complexity theory concern circuits over the Boolean domain. Besides their simplicity and the ease of comparing different languages, the actual architecture of computers is also an important motivating factor. On the other hand, by restricting attention to Boolean circuits, we lose sight of the much richer landscape of circuits over larger domains. Our goal is to bridge these two worlds: to use deep algebraic tools to obtain results in computational complexity theory, including circuit complexity, and to apply results from computational complexity to gain a better understanding of the structure of finite algebras.

In this paper, we propose a unifying algebraic framework which we believe will help achieve this goal. Our work is inspired by branching programs and nonuniform deterministic automata introduced by Barrington, as well as by their generalization proposed by Idziak et al. We begin our investigation by studying the languages recognized by natural classes of algebraic structures. In particular, we characterize language classes recognized by circuits over simple algebras and over algebras from congruence modular varieties.
\end{abstract}


\maketitle

\section{Introduction}
It is the fact that algebraic insights can lead to breakthroughs in computational complexity theory. Let us only mention a few examples: polynomials over finite fields and their usefulness in establishing $\mathsf{IP} = \mathsf{PSPACE}$ \cite{lund1992algebraic, shamir1992ip}; modelling program verification with  a Gr\"obner basis problems \cite{sankaranarayanan2004non, rodriguez2004automatic}; the group-theoretic proof of Barrington that languages from $NC^1$ can be recognized by bounded-width branching programs \cite{Barrington86}; currently explored machine learning algorithms heavily based on linear algebraic concepts \cite{lecun2015deep, vaswani2017attention}. One could enumerate so many more brilliant applications of algebra. There seem to be only a thin boundary between the world of Turing Machines and the world of Algebra.

Perhaps the most demonstrative example of that in the past decade was the P vs. NP-complete dichotomy theorem for  finite Constraint Satisfaction Problems (CSPs) \cite{bulatov2017dichotomy, zhuk2020proof}, which model a large class of natural algorithmic  problems. In the theory of CSP the key idea was to study not only the combinatorial nature of the constraints which define the problem but rather the interaction between the combinatorics and the symmetries of the constraints. These symmetries, called polymorphisms, are multiary functions which are closed under compositions and preserve the constraints. They form a structure called a termal clone. The field of Universal Algebra, which studies termal clones for about a century now, is what enabled this big CSP research project.

The other direction also holds and  the computational complexity notions can help to understand some of the algebraic concepts deeper. A prime example of that are Nonuniform Deterministic Finite Automata studied by Barrington, Straubing and Th\'erien \cite{barrington1990non}, where the expressive power of polynomial-size expressions over finite groups are linked to circuit complexity classes. Here DFA appear because, due to associativity, one may interpret these nonuniform expressions as acting on the state set of the finite automaton.
For instance, one can solve arbitrary problems in nonuniform $NC^1$ (including all the problems with very effective parallel algorithms) with polynomial-size expressions over any nonsolvable group. On the contrary, solvable groups seem to have much smaller computational power as expressions over such groups can only solve problems that are also solvable by bounded-depth circuits (in the class $CC^0$). Nilpotent groups provably have very small expressive power, as they essentially can only evaluate bounded-degree polynomials. Similar research was done for monoids, where expressions over solvable monoid seem to have more power than the ones of solvable groups \cite{barrington1988finite} (they solve problems in the class $ACC^0$).

Later it was spotted that the insights from \cite{barrington1990non} can be used to obtain polynomial or quasipolynomial time algorithms for the problem of solving equations in solvable groups/monoids \cite{barrington2000equation, GoldmannR02}. The problems of solving equations  was then considered in a much more general context of finite algebras \cite{IdziakK22}, which means finite universes with arbitrary finite sets of operations. Note that the class of these problems is very reach as it generalizes class of all finite CSPs (see \cite{IdziakK22}). The research was especially successful for the algebras from  so-called congruence modular varieties. This algebras generalize most of the structures studied in classical algebra, such as groups, rings, Boolean Algebra, Heyting Algebras, Lie Algebras, lattices, latin matrices and many more.
Recently some assumptions about circuit lower bounds have proven their usefulness in some variants of satisfiability problems for groups \cite{IdziakKKW22-icalp} as well as general algebras \cite{IdziakKK25} to provide randomized polynomial time algorithms.

In this paper we want to continue to study the connection between expressive power of circuits over finite algebras and classical Boolean circuit complexity classes. To define this correspondence correctly, we define how this multivalued circuits will be used to recognize languages. Say we have some finite universe $A$ and some set of operations $f_1, \ldots, f_k$ over $A$ that we will use as gates while building our circuit. This is our algebraic structure $\m A$.  Variables as well as constants from $A$ can appear on the input level of such circuits. This connects our consideration to so called polynomial clone of $\m A$, that is composition closed set of operations that contains constants from the universe. Let us say we have a language $L \subseteq \Sigma^*$ we want to recognize. Then we need a way to project the alphabet $\Sigma$ to the universe of our algebra, we need a way to interpret the result (we need the accepting set $S \subseteq A$) and because there are many possible input lengths we need not only one circuit, but a sequence of circuits over $\m A$ (that is why this model of computation is called nonuniform).

\begin{definition}\label{def-nudfa}
Non-uniform Deterministic Finite Automaton (NuDFA) over a finite algebra $\m A$ consists of:
\begin{itemize}
    \item a sequence $\set{\po t}_{i=1}^{\infty}$ of circuits over $\m A$ such that circuit $\po t_n$ is $n$-ary.
     \item function $\iota:\Sigma\mapsto A$,
    \item set $S\subseteq A$,
\end{itemize}
\end{definition}
\noindent A word $b_1\ldots b_n\in\Sigma^n$ gets accepted by such a NuDFA if \\
 $\po t_n(\iota(b_1),\ldots,\iota(b_n))\in S$.  A triple $(\po t, \iota, S)$ where $\po t$ is an $n$-ary circuit is called an $n$-ary program (in correspondence to straight-line programs). We will write $\po t[\iota](\overline{b})$ to denote $\po t(\iota(b_1),\ldots,\iota(b_n))$ and $\po t[\iota, S](\overline{b})$  as a characteristic  function of language recognized by program $(t,\iota,S)$. We say that NuDFA $(\set{\po t}_{i=1}^{\infty},\iota, S)$ is of polynomial size if there is polynomial $w$ such that size of $t_i$ is $O(w(i))$ for every $i$. The set of languages over binary alphabet recognized by polynomial size NuDFA's over $\m A$  we denote as $nu\mathsf{P}_{\m A}$. Note that the class $nu\mathsf{P}_{\m A}$ is closed under taking complement (as we can take the complement of the accepting set $S$) and bitwise complement (as we can take $\iota'(b) = \iota(\neg b)$). 

Note that in \cite{IdziakKK25} the word NuDFA was used only in reference to Barrington, Straubing and Th\'eriens' work, and in fact one should think of NuDFA over algebraic structures as of nonuniform circuit family rather than nonuniform DFA. Also, unlike \cite{IdziakKK25}, we fix $\iota, S$ for the circuit family. It is much simpler, more natural, and it allows us to capture certain complexity classes which would not appear otherwise. In particular monotone circuits can appear in our theorems only due to this change, and they better correspond to the structure of considered algebras.

Our notation $nu\mathsf{P}_{\m A}$ refers to nonuniform models of computation such as $\mathsf{P}/poly$, known as nonuniform $\mathsf{P}$, which in our notation is equal to $\palg{\m A}$ for $\m A = (\{0,1\}, \land, \lor, \neg)$.  $\mathsf{P}/poly$ is defined to be class of languages recognized by a family of polynomial size Boolean circuits using standard $\land, \lor, \neg$ operations. It is the highest possible class $\palg{\m A}$ can reach, as the $\land, \lor, \neg$  can simulate all the other operations on finite set. It was widely hoped in 70s and 80s that studying the class $\mathsf{P}/poly$ was the right way to approach a $\mathsf{P}\text{ vs.\! }\mathsf{NP}$ problem, as if one can prove $\mathsf{NP} \not\subseteq \mathsf{P}/poly$ it implies $\mathsf{P} \neq \mathsf{NP}$. As a natural proof barrier appeared on the scene \cite{razborov1994natural} and the progress slowed down, enthusiasm was a bit toned down, but still $\mathsf{P}/poly$ is considered one of the most important classes in the complexity zoo.

However lower bounds for $\mathsf{P}/poly$ seem to be far out of reach for the current techniques, in some restricted settings full success has been achieved.

\begin{example}\label{ex-lattice}
Let $\m A = (\{0,1\}, \land, \lor)$. Languages recognized by NuDFA's over $\m A$ fulfill one of the two following properties:
\begin{enumerate}
    \item\label{item-ex-first} If $(a_1,\ldots, a_n)\in L$, then for every $(b_1,\ldots, b_n)$ such that $a_i\le b_i$ holds, also $(b_1,\ldots, b_n)\in L$ holds.
    \item If $(a_1,\ldots, a_n)\in L$, then for every $(b_1,\ldots, b_n)$ such that $a_i\ge b_i$ holds, also $(b_1,\ldots, b_n)\in L$ holds.
\end{enumerate}
$nu\mathsf{P}_{(2,\wedge,\vee)}$ is equal to set of languages which characteristic function or its negation can be described by the sequence of  polynomial size circuits build of gates $\wedge$ and $\vee$.
\end{example}
\noindent Note that the study of monotone (in the sense of \cref{{item-ex-first}}) polynomial-size circuits has a big history with applications of some of the techniques in Proof Complexity \cite{krajivcek1994lower,  razborov1995unprovability,  bonet1995lower, pudlak1997lower}. Importantly, there are $\nptime$-complete languages that are monotone. With the right encoding $\mathsf{CLIQUE}$ is such a problem because adding an edge can produce only bigger cliques. By the classical result of Razborov $\mathsf{CLIQUE}$ requires superpolynomial-size monotone circuits \cite{Razborov1985}, which was later improved to show that the required size is actually exponential \cite{alon1987monotone}.  Finally it turns out that some monotone problems in $\ptime$ including perfect matching problem require superpolynomial-size circuits \cite{Razborov1985, tardos1988gap}, and the lower bound for matching was recently claimed to be lifted to exponential \cite{cavalar2025monotone}. So the negation really does a great job at making some of our algorithmic problems easier. All these results also apply to $\palg{\m A}$ from \cref{ex-lattice}.

From Post's theory \cite{Post1941} of algebraic clones on the two element domain one can see that there are only few possible polynomial clones, i.e. Boolean clone, lattice clone, clone generated by $\mathsf{xor}$, semilattice clones, unary clones. In contrast, there are uncountably many of polynomial clones on 3 element domain \cite{AgostonDH86}. Fortunately, we will only deal with countably many of them in this paper, as we consider algebras with finitely many generating operations.

In this vastness of structures, we can expect to find variety of complexity classes arising from them. In this paper we explore this space, on the one hand we study algebras from the congruence modular varieties, for which the tools of Universal Algebra such as Tame Congruence Theory and Commutator Theory work particularly well. We effectively propose a representation theory for them, which allows us to achieve the characterizations of $\palg{\m A}$.  On the other hand, we also study simple algebras for which we achieve almost full description with some interesting open cases (see \cref{sec-simple}). For finite algebras there is no theory of representation similar to the theory of groups or to Krohn-Rhodes Theory for monoids. Our research suggests that such representations might be possible to achieve, perhaps for bigger classes of structures than expected before. 

We will also see that many natural classes studied in computer science appear as $\palg{\m A}$ or as a limit of $\palg{\m A}$ for natural classes of structures. This provides an algebraic common view on the well-known complexity classes and provides yet another perspective on some of the open problems in the computational complexity theory.

\section{Algebraic Closure Properties}
In order to use algebraic tools in our research, the properties of the languages recognized by the automata must be related to the algebraic properties of the algebras used in the automata definition.  We will not go deeply into Universal Algebra here, but it is useful, also for the future proofs, to establish some basic facts here.

Subalgebra is a subset of an algebra closed under operations of this algebra together with the inherited operations. Set of languages recognized by programs over subalgebra is contained in the set of languages recognized by programs over the whole algebra.

\begin{fact}\label{fact-subalgebra}
Let $\m A$ be a finite algebra and let $\m B$ be its subalgebra.

Then, $nu\mathsf{P}_{\m B}\subseteq nu\mathsf{P}_{\m A}$.
\end{fact}

Normal subgroups and quotient groups are well known notions of group theory. Congruences and quotient algebras are their analogues in universal algebra. Formally, congruences are equivalence relations preserved by operation of the algebra. They can be described as all possible kernels for homomorphisms of $\m A$. As in the case of normal subgroups congruences form a lattice. Languages recognized by the quotient algebra are contained in the set of languages recognized by NuDFA's over the algebra.

\begin{fact}
Let $\m A$ be a finite algebra and let $\alpha$ be its congruence.

Then, $nu\mathsf{P}_{\m A/\alpha}\subseteq nu\mathsf{P}_{\m A}$.
\end{fact}

A little bit more complicated is situation in case of product of two algebras. Note that computing value of a circuit over a product $\m A\times \m B$ of two algebras $\m A$ and $\m B$ is like computing values for both algebras $\m A$ and $\m B$ separately. In case of programs the additional complication is that accepting set in NuDFA's over $\m A\times \m B$ can be not the product of  subsets of $A$ and $B$. 

\begin{fact}\label{fact-product}
Let $\m A$ and $\m B$ be finite algebras.

Then, 
\begin{itemize}
    \item $\palg{\m A}\subseteq \palg{\m A\times \m B}$.
    \item $\palg{\m B}\subseteq \palg{\m A\times \m B}$.
    \item if $L\in \palg{\m A\times \m B}$, then  there exist $L_1^{\m A},\ldots, L_k^{\m A}\in \palg{\m A}$ and $L_1^{\m B},\ldots, L_k^{\m B}\in \palg{\m B}$ such that $L=\bigcup_{i=1}^k L_i^{\m A}\cap L_i^{\m B}$.
    \item if $L_1\in\palg{\m A}$ and $L_2\in\palg{\m B}$\
    \begin{itemize}
        \item $L_1\cap L_2\in \palg{\m A\times\m B}$
        \item $L_1\cup L_2\in \palg{\m A\times\m B}$
    \end{itemize}   
\end{itemize}
\end{fact}
The above fact could suggest that the definition of $\palg{\m A}$ does not work well with the products. However, language classes that we obtain behave well under taking products, despite the fact that the definition of a program seems to be not designed to work well for them. 

%In view of \Cref{fact-subalgebra,fact-product} we can focus on subdirectly irreducible algebras i.e. algebras which cannot be presented as a subdirect product of non trivial algebras. 
 
\section{Results}\label{sec-intr-languages}
We will use the same name for the class of languages and the family of circuits that it defines. The meaning will be always clear from the context. 
It turns out that for many natural classes of finite algebras programs over them recognize natural classes of languages considered in circuit complexity. The widest class, we consider in this paper, is $\ppoly$, the class of languages recognized by boolean circuits of polynomial size. $\mathsf{NC}^2$ is a class of languages recognized by circuits of depth $O(log^2 n)$ and with polynomial number of gates, i.e.\! $2$-ary $\land, \lor$ and unary $\neg$. The last popular class we consider is $\mathsf{CC}^0$, the class of languages recognized by polynomial size and constant depth circuits consisting of $\mathsf{MOD}_m$ gates with unbounded fan-in. 

We call a language $\MON$ if it is recognized by polynomial size monotone circuits. All such languages fulfil \cref{item-ex-first} in \cref{ex-lattice}. If we fix an algebra such that all its operations preserve fixed total order, accepting sets $S$ being an upperset with respect to the same order and monotone  mapping $\iota$ then polynomial size NuDFA's of the form $(\set{\po t_i}_{i=1}^{\infty},\iota, S)$  recognize languages from $\MON$. $\BOUNDA{c}$ is a class of boolean circuits with number of inputs bounded by $c$. We define $\MMON$ as a class of languages defined by the circuits of the form  $\circuittwo{\BOUNDA{c}}{\MON}$. Note that the class $\MMON$ is closed under taking bitwise complements. It follows from the fact that monotone circuit with negated input is equivalent by De Morgan's laws to negation of another monotone circuit but with non-negated inputs.

Note that one can show that the language $CLIQUE$ requires superpolynomial size $\MMON$ circuits. It is because we can view these circuits as having only bounded number of negations close to the output (hidden in $\BOUNDA{c}$). Hence we can apply some classical results 
\cite{amano2005superpolynomial} to obtain separation. As we later characterize some of our classes $\palg{\m A}$ as having $\MMON$ circuits, this explicit $CLIQUE$  lower bound transfers to some of our algebras $\m A$.

It turns out that for every finite algebra $\m A$ from congruence modular variety $\palg{\m A}$ is equal $\ppoly$ or is contained in some, perhaps smaller, class of languages.

\begin{theorem}\label{thm-dichotomy}
Let $\m A$ be a finite algebra from a congruence modular variety.

Then, $\palg{\m A}=\ppoly$ or $\palg{\m A}\subseteq \circuittwo{\mathsf{NC}^2}{\mathsf{\MON}} $
\end{theorem}

In fact, we have much more precise result than \cref{thm-dichotomy} but we need some preparation to state the main theorem of the paper. To characterize languages recognized by different algebras from congruence modular variety we will use two important theories of universal algebra:  modular commutator theory (see \cite{fm}) and  tame congruence theory (\cite{hm}). Commutator theory generalize notion of commutator to universal algebraic realm, what enables generalization of such group-theoretical notions as abelianity, nilpotency and solvability for arbitrary algebraic structures. Tame congruence theory describes a local behaviour of finite algebras. Such a behaviour is connected with covering pairs of congruences $\alpha\prec\beta$ and so called $(\alpha,\beta)$-minimal sets. It turns out that such a minimal set can behave in one of five ways:
\begin{enumerate}
\item[{$\tn 1$}.]  a finite set with a (unary) group action on it,
\item[{$\tn 2$}.]  a finite vector space over a finite field,
\item[{$\tn 3$}.]  a two-element Boolean algebra,
\item[{$\tn 4$}.]  a two-element lattice,
\item[{$\tn 5$}.]  a two-element semilattice.
\end{enumerate}
Different prime quotients of the same algebra can behave in different ways. The typeset of an algebra $\m A$ i.e the set of prime quotients' types which can be found in $\m A$ is denoted by $\typset{\m A}$. Note that typesets of algebras from congruence modular varieties are contained in a set $\set{\tn 2, \tn 3, \tn 4}$.  Moreover in case of congruence modular varieties minimals sets of types $\tn 3$ and $\tn 4$ consist of exactly $2$ elements. In case of other types and other varieties the situation is more involved and will be clarified later. 

In a congruence modular variety, whenever all prime quotients of $\m A$ are of type $\tn 2$ the algebra $\m A$ is solvable and also all prime quotients of solvable algebras have only $\tn 2$ quotients \cite{hm, fm}. In a congruence modular variety all solvable algebras are Malcev, i.e.\! they have a ternary term $\po d$ satisfying $\po d(x,y,y) = \po d(y,y,x) = x$. 

We say that congruence $\alpha$ of $\m A$ is join irreducible if it is a join irreducible element of congruence lattice. It is well known that every join irreducible congruence $\alpha$ has its unique subcover $\alpha^{-}$. We define type of join irreducible congruence $\alpha$ as a type of quotient $(\alpha^-, \alpha)$. 

Subdirect product of algebras $\m A_1$,\ldots, $\m A_k$ is a subalgebra of their direct product such that projection on $i$-th coordinate is equal to $A_i$ for all $i$. An algebra is subdirectly irreducible if cannot be presented as a subdirect product of at least two nontrivial algebras. Every finite algebra $\m A$ is isomorphic to a subdirect product of a finite number of subdirectly irreducible finite algebras (which are homomorphic images of $\m A$ or in other words which are quotient algebras of $\m A$).

We say that that an algebra $\m A$ is totally ordered if there is total order preserved by all polynomial operations of $\m A$.

Now we can state the theorem. 

\begin{theorem}\label{thm-main}
Let $\m A$ be a finite %subdirectly irreducible 
algebra from congruence modular variety then:
\begin{enumerate}
    \item\label{item-main-1} If $\m A$ is nilpotent then $\palg{\m A}\subseteq \mathsf{CC}^0$.
    \item If $\typset{\m A} =\tn 2$, then $\palg{\m A}\subseteq \mathsf{NC}^2$.
    \item If $\tn 3\in\typset{\m A}$ then $\palg{\m A}=\ppoly$.
    \item  If $\m A$ is a subdirect product of totally ordered algebras then $\palg{\m A}=\MMON$.
    \item\label{cong-alp}  If there is a congruence $\alpha$ of $\m A$ such that $\typset{\m A/\alpha}=\set{\tn 4}$ and $\m A/\alpha$ is a subdirectly irreducible algebra which is not totally ordered then $\palg{\m A}=\ppoly$.
    \item\label{item-main-6} If $\m A$ has join irreducible congruences $\alpha<\beta$ such that type of $\alpha$ is $\tn 4$ and type of $\beta$ is $\tn 2$, then  $\palg{\m A} = \ppoly$.
    \item If none of Items \ref{item-main-1}-\ref{item-main-6} holds, then 
    $\palg{\m A} \subseteq \circuittwo{\mathsf{NC}^2}{\MON}$
\end{enumerate}
\end{theorem}

In fact, wherever $NC^2$ appears in the above theorem we can capture the real complexity class better, by a less natural class $\nuitdet$ of languages recognized by evaluating iterated determinant. We will explain the class later in the proper context.

Now we explain how to combine the results of the next sections 

\begin{proof}[Proof of \cref{thm-main}] 
If $\typset{A} = \set{\tn 2}$ by \cref{theorem-solv-nc2} we have $\palg{\m A} \subseteq NC^2$. If $\m A$ is additionally nilpotent $\palg{\m A} = CC^0$ by \cref{nilp-theorem}.

If $\tn 3 \in \typset{\m A}$ there is a two element set $U \subseteq A$ on which the algebra $\m A$ behaves like a Boolean algebra. So there are operations $\land_U, \lor_U, \neg_U$ in the polynomial clone of $\m A$ such that $(U, \land_U, \lor_U, \neg_U)$ is isomorphic to the two element Boolean algebra. Thus we can use $\iota$ to project to this two elements and encode arbitrary Boolean circuits by composing $\land_U, \lor_U, \neg_U$. Hence $\palg{\m A} = \ptime/poly$.  

If $\m A$ is subdirect product of totally ordered algebras, call them $\m B_1, \ldots, \m B_h$, every program $\progb{\po p}{\iota}{S}$ over $\m A$ can be simulated by a system of constant number of programs of the form $\progb{\po p^{\po B_i}}{\iota^{\po B_i}}{\{c^{\po B_i}\}}$, $\iota^{\po B_i}$ is just $\iota$ projected to the coordinate $i$ and $c^{\po B_i} \in B_i$. But by \cref{lm-totally-ordered} function computed by a program over each $\m B_i$ can be also computed by $\MMON$ circuit of size polynomial in the size of the program.
There is a bounded arity function $\po d$ which takes this programs from coordinates and reconstructs the answer for $\progb{\po p}{\iota}{S}$. Since $\MMON$ is closed under composing with bounded arity functions from the outside we see that $\palg{\m A} = \MMON$.

If there is a congruence $\alpha$ such as in (\ref{cong-alp}) by \cref{lm-subdirect-type-4} we have $\palg{\m A} = \ppoly$.
If $\m A$ has congruences $\alpha, \beta$ as in (\ref{item-main-6})  then by \cref{hard-combination-types} we have  $\palg{\m A} = \ppoly$. 

If none of the items 1-6 hold, then $\typset{A} = \set{\tn 2,\tn 4}$. Since \cref{item-main-6} does not hold by \cref{last-lem} there is a congruence $\beta$ such that $\typset{\m A/\beta} = \set{\tn 4}$ and all prime quotients below $\beta$ are of type $\tn 2$. Now $\m A/\beta$ has to decompose into a subdirect product of totally ordered algebras, otherwise \cref{cong-alp} would hold for some subdirectly irreducible quotient $(\m A/\beta)/\alpha$  of $\m A/\beta$ which is also a quotient of $\m A$. Then by \cref{final-mixed-th} $\palg{\m A} \subseteq \circuittwo{\mathsf{NC}^2}{\MON}$ which finishes the proof.

\end{proof}

\section{Background material}

\subsection{Universal algebra}

We use a standard universal algebraic notation (see \cite{BurrisS12}). Moreover, we use two important theories of universal algebra and their notions: tame congruence theory (TCT, see \cite{hm}) and modular commutator theory (see \cite{fm}).

 We use boldface capital letters to denote algebras and the same letters but using regular font to denote their universes. Congruences of an algebra $\m A$ form a complete lattice with the smallest element $0_{\m A}$ and the biggest element $1_{\m A}$, where  $0_{\m A}$ is the equality relation and $1_{\m A}$ is an equivalent relation with one equivalence class. Minimal non-trivial congruences i.e. congruences which cover $0_{\m A}$ are called atoms. It is well known that an algebra is subdirectly irreducible iff it has exactly one atom. Such the unique atom is called monolith. By polynomials over an algebra we mean every proper expression built of basic operations of the algebra, constants and variables. Polynomials define polynomial operations in an obvious way. Two algebras are polynomially equivalent if they have the same set of polynomial operations (or more precisely if one algebra is isomorphic to an algebra which has the same set of polynomial operations as the second algebra). For a subset $S\subseteq A$ an  algebra induced on $S$ is the algebra ($\m A|_{S}$) with universe $S$ and the set of basic operations consisting of all polynomial operations of $\m A$ which for values from $S$ return elements from $S$.

For a given covering pair of congruences $\alpha\prec\beta$ of $\m A$, an $(\alpha,\beta)$-minimal set is a minimal, with respect to inclusion, image of unary polynomial which do not collapse $\beta$ to $\alpha$. It is known that algebras induced on every two $(\alpha, \beta)$-minimal sets are polynomially equivalent. For a $(\alpha,\beta)$-minimal set $U$ and $a\in U$, $N=a/\beta\cap U$ such that $N\not=a/\alpha\cap U$ is called a trace. The sum of all traces of $U$ is a body and a tail of $U$ consists of all element of $U$ which are not contained in body i.e. $a\in U$ such that $a/\beta\cap U=a/\alpha\cap U$. For every minimal set $U$ there is polynomial $\po e_U$ such that $\po e(U)=U$ end $\po e_U(\po e_U(x))=\po e_U(x)$. One of the nice properties of algebras from congruence modular variety is the fact that its minimal sets have no tails. Every trace of $(\alpha,\beta)$-minimal set is, modulo $\alpha$, polynomially equivalent to the one of five types of local behaviour mentioned in  \cref{sec-intr-languages}. Every trace of the one minimal set is of the same type, and hence it makes sense to talk about the minimal set type. Minimal sets of types $\tn3$, $\tn 4$ end $\tn 5$ have exactly one trace. In case of type $\tn 2$ traces are polynomially equivalent to one-dimensional vector space and their sizes are of prime power order, this prime, the same for every trace in the one minimal set we call a characteristic of the minimal set of type $\tn 2$.

Let $k \in \mathbb{N}$ and $\m M$ be an algebra. Then the matrix power $\m M^{[k]}$ is defined in a standard way, i.e.\ the domain is just $M^{k}$ and operations on each component are just operations of $\m M$ but taking all variables from all components. That is if $m \in M^{k}$, and $m^i$ denotes the $i$-th component of $m$, the algebra $\m M^{k}$ contains all operations of the form $\po p(x_1, \ldots, x_n) = (\po p_1(x_1^1, \ldots, x_n^{k}), \ldots, \po p_k(x_1^1, \ldots, x_n^{k}))$ where $\po p_1, \ldots, \po p_k$ are some $n\cdot k$-ary polynomials of $\m M$. Matrix powers of minimal sets play central role in our investigation.  Actually, our results rely on the fact that every finite algebra from a congruence modular variety can be expressed as some type of a product of its minimal sets matrix powers. To be able  to state this fact properly we need to introduce the new type of a product.

\begin{definition}
Let $\m A, \m B_1, \ldots, \m B_h$ be finite algebras. We say that an algebra $\m A$ is an action product  $(\m B_1\boxtimes (\m B_2 \boxtimes (\ldots \boxtimes \m B_h)\ldots)$ whenever the domain $A = B_1 \times B_2 \times \ldots \times B_h$ and for each $n$-ary $\po f \in \pol{A}$ 
there are $\po f_1\ldots, \po f_h$, such that $\po f_i: (B_i \times \ldots \times B_h)^n \longrightarrow B_i$ and
$$\po f(x) = (\po f_1(x^{(1)}), \ldots, \po f_h(x^{(h)})),$$
where $x^{(i)}=(x_1^i,\ldots, x_n^{i},\ldots,x_1^h,\ldots, x_n^{h} )\in (B_i \times \ldots \times B_h)^n$. Moreover whenever we assign a constant value $c \in (B_{i+1} \times \ldots \times B_h)^n$ to $x^{(i+1)}$, the function $\po f_i(x_1^i,\ldots, x_n^{i}, c_1^{i+1},\ldots, c_n^{i+1},\ldots,c_1^h,\ldots, c_n^{h} )$ is an $n$-ary polynomial of the algebra $\m B_i$.
\end{definition}
The above definition may seem to be complicated but the main idea is quite simple. Algebras appearing in the product later act on their predecessor as it is in semidirect product of groups.

\subsection{Algebraic Branching Programs}
To understand the behaviour of circuits over solvable algebras, we need the notion of $\textit{algebraic branching program}$. We take our basic definitions from \cite{dvir2012separating}.
An algebraic branching program is a way to represent a formal polynomial over a field by an edge-labeled directed acyclic graph in which multiple edges between vertices are allowed. There is always one source vertex, and one sink vertex in the graph, and edges are labeled by either variable of a constant of the field $\mathbb{F}$. We will consider only $ABP$s over finite fields $\mathsf{GF}(p)$ where $p$ is a prime.  Each path from source to sink represents a monomial, created by multiplying variables/constants occurring on its edges. Then we say that $ABP$ represents/computes a polynomial $\po w$ equal to the sum of monomials on all such paths. 

The size of such $ABP$ is defined as the number of its vertices plus the number of the edges (but edges linearly dominate the size). Another way to define the size which is more useful in our context is by the following process, that we  later call ABP-process. We start with a single identity $1$ polynomial. Then we create new polynomials by applying one of the two steps:
\begin{enumerate}
    \item take an already created polynomial and multiply it by a constant or a variable
    \item add $n$ already created polynomials.
\end{enumerate}
One can see that this two notions of size are linearly-dependent on each other, i.e. if some $ABP$ represents a polynomial $\po f$ using $l$ edges then there is $O(l)$ step process defining $\po w$, and also from each $l$-step process we get $ABP$ of size $O(l)$ in an obvious way. Indeed, the multiplication step corresponds to adding edge to a DAG and addition corresponds to introducing new vertex in the DAG. As we are only interested in asymptotic sizes, we use this two notions of size interchangeably. If a polynomial has ABP of size $O(l)$, we will sometimes say that $\po w$ is of ABP complexity $O(l)$.

Determinantal complexity of a polynomial $\po w$ over field $\mathbb{F}$ is the smallest size of a matrix filled with variables and constants from $\mathbb{F}$ such that $\po w$ is (as a polynomial) equal to determinant of the matrix. It is known, that every polynomial can be represented as determinant of such matrix, moreover the determinantal complexity of polynomials represented by small arithmetic formulas is also relatively small \cite{valiant1979completeness}. We will use the following deep result.

\begin{theorem}[\cite{valiant1979completeness, malod2008characterizing, berkowitz1984computing, samuelson1942method, mahajan1997determinant}]\label{abp-to-det}
Determinantal complexity of a polynomial and the smallest size of ABP representing it are polynomially related.
\end{theorem}

We will also use the fact that we can make affine substitutions for variables in a given ABP and it does not change its size too much. We note that many authors define ABP so that in the DAG defining it one can put affine combinations on edges instead of just variables and constants, but for our definition the following fact allows us to manage affine substitutions.

\begin{fact}\label{affine-subst}
Let $\po w$ be a polynomial  over $n$ variables represented by an ABP of size smaller than $l$. Let $A_1, A_2, \ldots, A_n$ be affine combinations of this $n$ variables. Then the polynomial $\po w(A_1, \ldots, A_n)$ has ABP of size  $O(nl)$. Additionally if $A_i$ depend only on one variable the size is $O(l)$. 
\end{fact}
\begin{proof}
   After such affine substitution label on each edge becomes an affine combination of variables $\alpha_1 \cdot x_1 + \ldots + \alpha_n x_n +c$. Its enough to replace this edge by an ABP computing this affine combination, i.e.\! put a source of the affine ABP in the beginning of the edge and sink of the ABP in the end of the edge. This  ABP has size bounded by $O(n)$ and its $O(1)$ when such affine combination has one variable.
\end{proof}

Another fact we will use is that we can multiply polynomials represented by ABPs without expanding the size of the representation too much.

\begin{fact}\label{abp-mult}
Let $\po w_1, \po w_2$ be two polynomials represented by ABPs of size at most $l_1$ and $l_2$ respectively. Then the polynomial $\po w_1 \cdot \po w_2$ has ABP of size at most $l_1 + l_2$.
\end{fact}
\begin{proof}
    In the DAG definition of ABP, we indentify the source of ABP for $\po w_2$ with a sink of ABP for $\po w_1$  and we get an ABP for $\po w_1 \cdot \po w_2$.
\end{proof}

\subsection{Circuit Combinatorics}
In this paper we will consider bounded depth Boolean circuit families with the gates of unbounded fan-in. As we consider only decision problems, we will always have only one output gate. As opposed to usual notions in circuit complexity, gates of our circuits will have different sizes, as they will sometimes represent complex computational objects. By default gates have size one unless specified directly. 
We measure size of the circuit as the total number of its edges plus the sum of the sizes of its gates. 

 Gate of type $\MODG{m}$ is labelled with an accepting set $S \subseteq \mathbb{Z}_m$. It takes an arbitrary number of Boolean inputs and sums them modulo $m$. The gate returns $1$ iff the sum belongs to $S$. This gate is always of size one. The $CC_h[m]$ circuit is a circuit which has $h$ layers with $\MODG{m}$ gate on each layer. 

 Gate $\ANDG{d}$ computes arity $d$ conjunction of its input wires. We will mention the complexity class $\mathsf{AC}_h[m]$. This are languages recognized by circuits of height $h$ using gates $\ANDG{}$, $\ORG{}$, $\MODG{m}$ of unbounded fan-in and the $\neg$ gate.
 
 Gate of type $\PROGG{\m A}$ is labeled with an $n$-ary program $\progb{\po p}{\iota}{s}$ over $\m A$. It takes boolean inputs $b_1, \ldots, b_n$ and returns $\progb{\po p}{\iota}{s}(b_1, \ldots, b_n)$. The size of such gate is the size of the circuit defining $\po p$ in $\m A$. Gate of type $\BABPG{p}$ is labelled with $n$-variate $ABP$ over $GF(p)$ and the accepting set $S \subseteq GF(p)$. It take $n$ Boolean inputs and returns one if evaluation of the polynomial corresponding to this $ABP$ on these inputs belongs to set $S$. The size of such gate is just the size of the $ABP$. 

We will do inductive proofs on both the structure of the algebra $\m A$ as well as the definition of the circuits/polynomials defined over $\m A$, for this reason we introduce the following notation.
\begin{definition}
For a circuit $\po p$ over an algebra $\m A$ and its gate $G$ we write $\po p_G$ for a unique subcircuit of $\po p$ which consists  output gate $G$ and all its ancestors in $\po p$. By $U(G)$ we mean the set of all  ancestors of G in circuit $\po p$.

We write $\po f_G$ for a basic operation computed by the gate $G$.
\end{definition}

%type 3 and type 2 over type 4

\section{Local lattice behaviour}
In this section we consider languages recognized by NuDFA's over algebras with typeset equal to $\set{\tn 4}$. First we define binary relation which plays a central role in our proofs.
\begin{definition}
Let $\m A$ be a finite algebra, let $\alpha \prec \beta$ be congruences of $\m A$, and let $U = \{0_U, 1_U\}$ be a $(\alpha,\beta)$-minimal set of type $\tn 4$ with no tail. Let $e_1,e_2,\ldots,e_k$ be all polynomials of $\m A$ whose ranges are equal to $U$. 

\[
\le_U{\defeq}\set{(a,b)\in \beta : e_i(a)\le e_i(b)\text{ for } i=1,\ldots,k}.
\]
\end{definition}

It turns out that $\le_U$ is a preorder preserved by all operation of the algebra.

\begin{lemma}\label{lm-order}
Let $\m A$ be a finite algebra, let $\alpha\prec\beta$ be congruences of $\m A$, and let $U$ be a $(\alpha,\beta)$-minimal set of type $\tn 4$ with no tails. Then:
\begin{enumerate}
    \item $\le_U$ is a preorder;
    \item\label{lm-order-item-2} $\le_U/\alpha$ is an order on the $\beta/\alpha$-cosets of $\m A/\alpha$;
    \item\label{lm-order-item-3} every polynomial of $\m A$ preserves $\le_U$.
\end{enumerate}
\end{lemma}
\begin{proof}
Let $e_1,e_2,\ldots,e_k$ be the polynomials of $\m A$ with range equal to $U$. By definition, $\le_U$ inherits reflexivity and transitivity from the coordinatewise order on $2^k$, and hence $\le_U$ is a preorder.

By \cite[Theorem~2.8(4)]{hm}, for every pair $(a,b)\in \beta\setminus\alpha$ there exists a unary polynomial $\po f$ such that $\po f(A)=U$ and $(\po f(a),\po f(b))\in \beta\setminus\alpha$. Hence the mapping
\[
x \mapsto(e_1(x),\ldots,e_k(x)),
\]
for fixed $a$, is an embedding of the $\beta/\alpha$-coset of $a/\alpha$ into $2^k$. This proves \cref{lm-order-item-2}.

To prove \cref{lm-order-item-3}, assume for contradiction that there exists an $n$-ary polynomial $\po p$ of $\m A$ and tuples $\overline a,\overline b\in A^n$ such that $a_i\le_U b_i$ for all $i$, but
\[
\po p(a_1,\ldots,a_n)\not\le_U \po p(b_1,\ldots,b_n).
\]
Replacing, step by step, $a_i$ with $b_i$, we will obtain at some point $\overline{a'}$ and $\overline{b'}$ which differ on one position only such that  $a'_i\le_U b'_i$ for every $i$ and $\po p(a'_1,\ldots,a'_n)\not\le_U\po p(b'_1,\ldots,b_n')$. Hence, we may assume that $\po p$ is unary and that $a\le_U b$ while $\po p(a)\not\le_U \po p(b)$. Then there exists $i$ such that
\[
e_i(\po p(a))>e_i(\po p(b)).
\]
Since $e_i\circ \po p$ is a polynomial with range $U$, there exists $j$ such that $e_j=e_i\circ \po p$. Consequently,
\[
e_j(a)>e_j(b),
\]
contradicting $a\le_U b$.
\end{proof}
%Points $1,2$ follow from the fact that we project $A$ to the set $U$ on which we have a natural order. Polynomials of $\m A$ preserve $\le_U$  because if they wouldn't $e_1, \ldots, e_k$ would not be all polynomials whose ranges are equal to $U$.

If $\m A$ is totally ordered then the class of languages recognized by NuDFA's over $\m A$ is quite restricted.

\begin{lemma}\label{lm-totally-ordered}
Let $\m A$ be a totally ordered algebra. Then,
\[
\mathsf{nuP}_{\m A} \subseteq \MMON.
\]
\end{lemma}
\begin{proof}
Assume without loss of generality that $A\subseteq\set{0,1}^k$ for some $k$ and that total order on $A$ is just  the coordinatewise order on $\set{0,1}^k$. Let $\pi_i$ be a function which project element of $A$ onto its $i$-th coordinate. Note that since $\m A$ is totally ordered, for every basic operation of $\m A$ its projection on any coordinate can be expressed  as a circuit over $(2,\wedge, \vee)$ which takes arguments from coordinates as inputs. Hence, projection of every polynomial size circuit of $\m A$ can be expressed as a polynomial size circuit over two element lattice. The same is true for characteristic function of every upper set of $A$ except the size which, in this case, is bounded by a constant. 

First, we will show that polynomial size NuDFA $(\set{\po t_i}_{i=0}^{\infty},\iota, S)$ with monotone $\iota$ and $S$ being upper set of $A$ recognize language contained in $\MON$. In such a case $\iota(0)\le\iota(1)$ and hence $\pi_i(\iota(0))\le\pi_i(\iota(1))$ for all $i$. Hence, for all $i$ either $\pi_i(\iota(0))=\pi_i(\iota(1))$ or $\pi_i(\iota(0))=0$ and $\pi_i(\iota(1))=1$. Therefore, for every $i$,  $\pi_i(\po t_n(\iota(b_1),\ldots,\iota(b_n)))$ can be expressed as a polynomial size circuit over two element lattice on inputs $b_i$. As a consequence, $\po t_n[\iota, S](\overline{b})$ is a composition of monotone circuits computing\break $\pi_i(\po t_n(\iota(b_1),\ldots,\iota(b_n)))$, for ever $i$, and the characteristic function of $S$. 

Now, we will show that ever language recognized by NuDFA over $\m A$ is contained in $\MMON$. Note, that for ever $S\subseteq A$ it holds that $S=\bigcup_{s\in S}(s\mathord{\uparrow}\setminus \bigcup_{a>s}a\mathord{\uparrow})$, where $s\mathord{\uparrow}$ is the upper set $s\mathord{\uparrow}=\set{a\in A:a\geq s}$. Hence, since every set is a result of finitely many set operations on upper sets, the characteristic function of a language recognized by polynomial size  NuDFA over $\m A$ with monotone function $\iota$ and arbitrary set $S$ can be computed by composition of monotone circuits and bounded arity Boolean circuit i.e  $\BOUNDA{c}\circ\MON\subseteq\MMON$.  Finally, observe that if $\iota$ is not monotone, then $\iota'(x)=\iota(\neg x)$ is, and polynomial size NuDFA $\m D'=(\set{\po t}_{i=1}^{\infty},\iota', S)$ recognizes language contained in $\MMON$. Since, language recognized by $\m D=(\set{\po t}_{i=1}^{\infty},\iota', S)$ is just a bitwise complement of language recognized by $\m D'$ and $\MMON$ is closed under taking  bitwise complement, it  follows that language recognized by $\m D$ is contained in $\MMON$. This observation completes the proof.
\end{proof}
Finally, we are ready to prove the main result of this section.
\begin{lemma}\label{lm-subdirect-type-4}
Let $\m A$ be a subdirectly irreducible algebra with $\typset{\m A}=\set{\tn 4}$, monolith $\mu$ and such that $(0_{\m A},\mu)$-minimal sets have no tails.
If there exists a total order preserved by all operations of $\m A$, then
\[
\mathsf{nuP}_{\m A} \subseteq \MMON.
\]
Otherwise,
\[
\mathsf{nuP}_{\m A} = \textsc{P}/\text{poly}.
\]
\end{lemma}

\begin{proof}
Let  $U=\set{0_U,1_U}$ be a $(0_{\m A},\mu)$-minimal set of $\m A$. Let $\po e_1,\po e_2,\ldots,\po e_k$ be the  polynomials of $\m A$ with range equal to $U$. Assume without loss of generality that $\po e_1$ is an idempotent operation.  We will show that if $\palg{\m A}\neq \textsc{P}/\text{poly}$, then $\le_U$ is a total order.

First note that for every $a,b\in A$ we have $\mu\subseteq \Theta(a,b)$. Hence there exists a sequence $\{\po f_1,\ldots,\po f_l\}\subseteq \poln{1}{A}$ such that
\[
0\in\{\po f_1(a),\po f_1(b)\},\quad
1\in\{\po f_l(a),\po f_l(b)\},
\]
and
\[
\{\po f_i(a),\po f_i(b)\}\cap \{\po f_{i+1}(a),\po f_{i+1}(b)\}\neq\emptyset
\]
for all $i=1,\ldots,l-1$.

Thus there exists $i$ such that $\po e_1(\po f_i(a))\neq \po e_1(\po f_i(b))$, and consequently there exists $j$ such that $\po e_j=\po e_1\circ \po f_i$. It follows that for all distinct $a,b\in A$,
\[
(\po e_1(a),\ldots,\po e_k(a))\neq (\po e_1(b),\ldots,\po e_k(b)),
\]
and hence $\le_U$ is an order on $A$.

Assume now that there exist $a,b\in A$ such that neither $a\le_U b$ nor $b\le_U a$. Then there exist $i,j$ such that
\[
\po e_i(a)< \po e_i(b)\quad\text{and}\quad \po e_j(a)>\po e_j(b).
\]
Let $C(b_1,\ldots,b_n)$ be a Boolean circuit. By De Morgan's laws, we may assume that all negations occur only at the input level. Replace the Boolean operations $\land,\lor$ by the operations $\land_U,\lor_U$ on $U$ to obtain a polynomial $\po p$ of $\m A$. If the $m$-th input is negated, replace it by $\po e_j(x_m)$; otherwise replace it by $\po e_i(x_m)$. Let the projection $\iota$ be such that $\iota(0)=a$ and $\iota(1)=b$, and let the accepting set be $S=\set{1_U}$.

The resulting program computes the same Boolean function as $C$, and its size is linear in the size of $C$. Hence $\palg{\m A}=\ppoly$.
\end{proof}

\section{The structure of finite algebras}
In this section we show a structural results giving description of algebras from congruence modular variety.  We strongly believe that these results are of independent interest and can be generalised to work for arbitrary finite algebras. 

\begin{lemma}\label{lm-decomposition}
Let $\m A$ be a finite algebra from a congruence modular variety, let $\alpha\in\con{\m A}$ be an atom, and let $U$ be a $(0_{\m A},\alpha)$-minimal set.

Then there exists $k$ such that $\m A$ is isomorphic to a subreduct of $U^{[k]}\boxtimes \m A/\alpha$.
\end{lemma}

\begin{proof}
Note that for every $(a,b)\in\alpha$ there exists $f\in \poln{1}{A}$ such that $f(\m A)\subseteq U$ and $f(a)\neq f(b)$ (\cite[Theorem~2.8(4)]{hm}). Let $\set{\po e_1,\ldots,\po e_k}$ be the set of all such unary polynomials. It is easy to see that for $a,b\in A$ we have
\[
a=b \iff (\po e_1(a),\ldots,\po e_k(a),a/\alpha)=(\po e_1(b),\ldots,\po e_k(b),b/\alpha).
\]

We will show that $\m A$ is isomorphic to a subreduct of $U^{[k]}\boxtimes \m A/\alpha$ via the mapping
\[
h(x)=(\po e_1(x),\ldots,\po e_k(x),x/\alpha).
\]
It suffices to show that for every basic $n$-ary operation $\po f^{\m A}$ of $\m A$ there exists an operation $\po f^{U^{[k]}\boxtimes \m A/\alpha}$ such that
\[
h(\po f^{\m A}(x_1,\ldots,x_n))
=
\po f^{U^{[k]}\boxtimes \m A/\alpha}(h(x_1),\ldots,h(x_n)).
\]

Fix $a_1,\ldots,a_n\in A$. It is enough to show that for each $i$ there exists a polynomial $\po p_i$ of $U$ such that
\[
\po e_i(\po f(x_1,\ldots,x_n))
=
\po p_i(\po e_1(x_1),\ldots,\po e_1(x_n),\ldots,\po e_k(x_1),\ldots, \po e_k(x_n))
\]
for all $x_1\in a_1/\alpha,\ldots,x_n\in a_n/\alpha$.

We split the rest of the proof according to the TCT type of $U$. Since $\m A$ belongs to a congruence modular variety, the only possible types of $U$ are $\tn 2$, $\tn 3$, and $\tn 4$.

\paragraph{Type $\tn 2$.}
In this case $\alpha$ is abelian, i.e.\ $[\alpha,\alpha]=0_{\m A}$. Hence a difference term $d$ defines ternary abelian group operations on $\alpha$-cosets (see \cite[Lemma~5.6]{fm}) and commutes with every polynomial of $\m A$ on $\alpha$-related arguments (see \cite[Proposition~5.7]{fm}). Moreover, the difference term induces a group operation on the traces of $U$, since
\[
\po d(\po e_U(x),\po e_U(y),\po e_U(z))=\po e_U(\po d(x,y,z))
\]
for all $x,y,z\in U$.

Let $\set{0_{a/\alpha}:a\in A}$ be a transversal of $\alpha$ such that $0_{u/\alpha}\in U$ for every $u\in U$. For $x,y\in a/\alpha$ define $x+y=\po d(x,0_{a/\alpha},y)$. Let $\po f$ be an $n$-ary basic operation of $\m A$, fix $a_1,\ldots,a_n\in A$, and define
\[
\po g_i(x_1,\ldots,x_n)
=
\po e_i(\po f(x_1,\ldots,x_n))
-
\po e_i(\po f(0_{a_1/\alpha},\ldots,0_{a_n/\alpha})).
\]
Let $0=0_{f(a_1,\ldots,a_n)/\alpha}$. Then $g_i(0_{a_1/\alpha},\ldots,0_{a_n/\alpha})=0$. 

Observe that for $x_1\in a_1/\alpha$, \ldots, $x_n\in a_n/\alpha$ 
\begin{align*}
\po g_i(x_1,x_2,\ldots, x_n)=\\=\po g_i(\po d(x_1,0_{a_1/\alpha},0_{a_1/\alpha}),\po d(0_{a_2/\alpha},0_{a_2/\alpha},x_2),\ldots,\po d(0_{a_n/\alpha},0_{a_n/\alpha},x_n))=\\=
\po d(\po g_i(x_1,0_{a_2/\alpha},\ldots, 0_{a_n/\alpha}),\po g_i(0_{a_1/\alpha},0_{a_2/\alpha},\ldots, 0_{a_n/\alpha}),\\\po g_i(0_{a_1/\alpha},x_2,\ldots, x_n))=\\=
\po g_i(x_1,0_{a_2/\alpha},\ldots, 0_{a_n/\alpha})+\po g_i(0_{a_1/\alpha},x_2,\ldots, x_n).
\end{align*}
In this way we obtain that for $x_j\in a_j/\alpha$ 
\[
\po g_i(x_1,\ldots,x_n)
=
\sum_{j=1}^n
\po g_i(0_{a_1/\alpha},\ldots,0_{a_{j-1}/\alpha},x_j,0_{a_{j+1}/\alpha},\ldots,0_{a_n/\alpha}).
\]
Each summand is either constant or a unary polynomial of $\m A$ with image contained in $U$, and hence equals $\po e_{s_j}(x_j)$ for some index $s_j$. Consequently,
\[
\po g_i(x_1,\ldots,x_n)=\sum_{j\in J} \po e_{s_j}(x_j),
\]
and therefore
\[
\po e_i(\po f(x_1,\ldots,x_n))
=
\po p_i(e_1(x_1),\ldots,e_k(x_n))
\]
for some polynomial $\po p_i$ of $U$.

\paragraph{Type $\tn 3$.}
In this case the claim is immediate, since every function on a two-element set is a polynomial of $U$.

\paragraph{Type $\tn 4$.}
Note that, since $\alpha$ is an atom,  $\le_U$, by \cref{lm-order}, is an order on $\alpha$-cosets. Moreover,  for any $n$-ary polynomial $\po f$ of $\m A$ and fixed $a_1,\ldots,a_n\in A$, the restriction
\[
\po f|_{a_1/\alpha\times\cdots\times a_n/\alpha}
\]
is monotone with respect to $\le_U$. Since, every monotone Boolean function is a lattice polynomial, each function $\po e_i\circ \po f$ restricted to these cosets can be expressed as a lattice polynomial in the variables $\po e_j(x_l)$, for $j=1,\ldots,k$ and $l=1,\ldots,n$. This completes the proof.
\end{proof}

In the case of atom of type $\tn 2$ we can give even more precise description. 
\begin{lemma}\label{struct-alp-cong}
Let $\m A$ be a finite algebra from a congruence modular variety, let $\alpha\in\con{\m A}$ be an atom, and let $U$ be a $(0_{\m A},\alpha)$-minimal set of type $\tn 2$.

Then there exists $k$ such that $\m A$ is isomorphic to a subreduct of $\mathbb{Z}_p^{[k]}\boxtimes \m A/\alpha$, where $p$ is a prime divisor of $|U|$.
\end{lemma}
We omit a straightforward proof of \cref{struct-alp-cong}. Let us only remark that the proof can be divided into two steps. First we can show that instead of projections into minimal set, used in the proof of \cref{lm-decomposition},  we can project elements of the algebra into one of the minimal set traces. Then, it is enough to show that every trace is subreduct of matrix power of $\mathbb{Z}_p$.

\section{Solvable algebras}
Now we focus on finite algebras $\m A$ which are solvable, i.e.\! they satisfy $\typset{\m A} \subseteq \{2\}$. We are going to characterize $\palg{\m A}$ with the use of algebraic branching programs. It is quite unexpected, that Barrington used nonsolvable groups to demonstrate the power of (Boolean) branching programs, and here we use algebraic branching programs to bound the power of solvable algebras and their circuits.

 The following lemma follows from analysis of the proof of Lemma 9.6 in \cite{IdziakKK25}.

\begin{lemma}{\label{coll:simaff}}
Let $\m A$ be a finite simple algebra of affine type of characteristic $p$. Let $\progb{\po p}{\iota}{S}$ be a program over $\m A$ of length $l$. Then there is an integer $k$ depending only on $\m A$ and a $\BOUNDG{k} \circ \MODG{p}$ circuit of size $O(l)$ computing the same function as $\po p$ does.  
\end{lemma}

To precisely characterize $\palg{\m A}$ for nonsimple $\m A$ we introduce the following complexity classes. 

\begin{definition}
    We say that a language $L$ is in the class $\nudet{p_1} \circ \ldots \nudet{p_h}$ whenever it is recognized by a sequence of polynomial size $\BABPG{p_1} \circ \ldots \circ \BABPG{p_h}$ circuits. In particular $\nudet{p}$ is a set of languages recognized by $\BABPG{p}$ gates of the size polynomial in the input legth.
    
    $\nuitdet$ is the set sum of classes $\nudet{p_1} \circ \ldots \nudet{p_h}$ for all sequences of primes $p_1, \ldots, p_h$.
\end{definition}

Notice that for a polynomial $\po w$ of ABP complexity $l$ from \cref{abp-to-det} we can replace evaluation of expression $\po w(\ov b) \in S$ by computing determinant of a matrix of size $O(\text{poly}(l))$ and checking if a result is in $S$. It justifies the name of the class $\nudet{p}$ as the $\BABPG{p}$ gates just evaluate such polynomials. Moreover, computing such determinant can be famously done uniformly in the class $NC^2$. Since we have different gates for different input lengths the resulting class is nonuniform.

\begin{fact}\label{solv-nc_2}
    $\nudet{p} \subseteq \text{(nonuniform) }NC^2$.
\end{fact}

On the other hand, every $n$-ary formula in $NC^1$ of size $\text{poly}(n)$ can be converted to arithmetic formula over $GF(p)$ of height $O(\log n)$, then expanded to arithmetic expression of size $O(\text{poly}(n))$ and then written as a determinant of a $O(\text{poly}(n))$ size matrix by \cite{valiant1979completeness}. Hence, by \cref{abp-to-det} $NC^1 \subseteq \nudet{p}$.

In our proofs we will use the following

\begin{fact}\label{compose_poly_const}
Let $p$ be a prime and $k$ be a fixed constant. Let $\po w_1, \po w_2, \ldots, \po w_k$ be $n$-ary  polynomials over the field $GF(p)$ which can be represented by ABPs of size $O(l)$. Let $\po d$ be a $k$-ary function of type $GF(p)^k \longrightarrow GF(p)$. Then the function $GF(p)^n \longrightarrow GF(p)$ induced by expression $\po d(\po w_1, \ldots, \po w_k)$  can be represented by a polynomial that has ABP of size $O(l)$.
\end{fact}
\begin{proof}
Function $\po d$ is of bounded arity, so it can be represented by a sparse polynomial over $GF(p)$ which in its definition requires some constant number of additions and multiplications. So by consecutively applying \cref{abp-mult} to every multiplication in this polynomial, 
the composition $\po d(\po w_1, \ldots, \po w_n)$ can be represented by a polynomial of size $O(l)$.
\end{proof}

It turns out that the class $\nudet{p}$ can be exactly captured as $\palg{A}$ for some solvable algebra $\m A$. 

\begin{theorem}\label{nudet-exact}
For each prime $p$ there is a solvable Malcev algebra $\m A$ such that polynomial size programs over $\m A$ recognize exactly languages from $\nudet{p}$.   
\end{theorem}

\begin{proof}[Proof sketch]
    
Let the underlying set of $\m A$ be $Z_p \times Z_p$ and let $\m A$ have $5$ operations. 
\begin{enumerate}
    \item binary abelian $+$ opearion of a group $Z_2 \times Z_2$
    \item unary projections $\pi_1, \pi_2$ satisfying\\ $\pi_1(x,y) = (x,0), \pi_2(x,y) = (0,y)$
    \item unary $u(x,y) = (y,0)$
    \item binary quasiprojection $v((x_1, y_1), (x_2, y_2)) = (x_1 \cdot y_2, 0)$
\end{enumerate}
It is automatic to check with the definitions that this algebra is Malcev and solvable. Moreover one can check by a simple inductive argument that every polynomial $\po p(x_1, \ldots, x_n)$ of $\m A$ of length $l$ is of the form $(\po w(x_1^1, \ldots, x_n^1, x_1^2, \ldots, x_n^2), \po a(x_1^2, \ldots, x_n^2))$, where $\po w(x_1^1, \ldots, x_n^1, x_1^2, \ldots, x_n^2)$ is a polynomial over $GF(p)$ with $ABP$ complexity $O(l)$ and $\po a$ is an affine function over variables $x_i^2$. Indeed all what the basic operations on the first coordinate allow is just addition and multiplication by an affine function from the second coordinate. The multiplication by an affine function can be simulated by $O(n)$ ABP steps. On the other hand any polynomial $\po w(x_1^2, \ldots, x_n^2)$ over $GF(p)$ with some bounded ABP complexity can be obtained on the first coordinate by applying polynomially many basic operations of $\m A$. This two facts together would give us $\palg{\m A} = \nudet{p}$, the problem is that while rewriting program to polynomial over $GF(p)$ we have to deal with $\iota, S$. To handle/eliminate $\iota$ we interpolate it with an affine function and apply \cref{affine-subst}. To emulate $S$  we apply \cref{compose_poly_const}. Another issue is that $\m A$ has two coordinates, but we can also conjunct them by an application of \cref{compose_poly_const}. Rewriting polynomial over $GF(p)$ to a proper program of our algebra is straightforward.
\end{proof}

Now we show how to inductively upperbound the complexity class $\palg{A}$, when the induction is on the height of congruence lattice of $\m A$. 

\begin{lemma}\label{solvable-induction}
Let $\m A$ be a finite algebra from a congruence modular variety and $\alpha$ be its atom of type $2$ and characteristic $p$. Then any function computed by a program over $\m A$ of size $l$ can be also computed by circuit of type $\BABPG{p} \circ \PROGG{\m A/\alpha}$ of size $O(\text{poly}(l))$.
\end{lemma}
\begin{proof}
 Let $\m C = \m A/\alpha$. By \cref{struct-alp-cong} the algebra $\m A$ is isomorphic to a subreduct of $\mathbb{Z}_p^{[k]}\boxtimes \m C$. Without loss of generality we assume that $\m A$ is a reduct of $\mathbb{Z}_p^{[k]}\boxtimes \m C$ as taking subalgebras decreases the power of the programs. For every $a \in A$ we write $a^{\m C}$ for its coordinate corresponding to $\m C$, and for $s\in \{1, \ldots, k\}$ we write $a^{s}$ for its coordinate corresponding to the $s$-th component of $\mathbb{Z}_{p}^k$.   

 Let $\po f(x_1,\ldots, x_r)$ be a basic operation of $\m A$. Then
 by the definition of matrix product $\po f(x_1,\ldots, x_r)^{s}$ depends on the  $r + r \cdot k$ values
 \[x_1^{\m C}, \ldots, x_r^{\m C}, x_1^1, \ldots, x_r^k
 \]
 in such a way that after fixing $x_1^{\m C}, \ldots, x_r^{\m C}$ it becomes a polynomial of $(Z_p, +)$, that is an affine function. It means that 
\[
 (\po f(x_1, \ldots, x_r))^s = \sum_{c_1,..,c_r \in \m C} \chi_{c_1}(x_1^{\m C})\cdot \ldots \cdot \chi_{c_r}(x_r^{\m C}) \cdot A_{\overline{c}}(x_1^1,..,x_r^k)
\]
where $\chi_c(x)$ is a unary characteristic function $C \longrightarrow Z_p$ which is equal to 1 when $x=c$ and is equal to $0$ otherwise, and each $A_{\overline{c}}$ is an affine function coming from substituting $c_1,..,c_r$ for the respective variables $x_1^C, \ldots, x_r^C$. But then we can regroup formula  according to the variables $x_i^j$ to get
\begin{align*}
 (\po f(x_1, \ldots, x_r))^s = \sum_{i=1..r}\sum_{j=1..k} x_i^j \cdot &\po w_{i,j}(\chi_{c_1}(x_1^{\m C}), \chi_{c_2}(x_1^{\m C}), ..,\chi_{c_m}(x_r^{\m C}))\\
 +&\po w(\chi_{c_1}(x_1^{\m C}), \chi_{c_2}(x_1^{\m C}), ..,\chi_{c_m}(x_r^{\m C}))
\end{align*}
where $c_1,\ldots c_m$ enumerate all elements of $\m C$ and each $\po w_{i,j}$ as well as $\po w$ is a $rm$-ary degree $r$ polynomials over $GF(p)$. 

Now take arbitrary $n$-ary polynomial/circuit $\po p(x_1, \ldots, x_n)$ of $\m A$ which is built of $l$ gates. From now, without loss of generality, we will treat  $l$ as a size of $\po p$, because its smaller than the actual size of $\po p$ (it does not account for the edges). We will prove by induction on the definition of $\po p$ that there exists $nk+ml$-ary polynomial $\po v_{\po p}^s$ over $GF(p)$ which has ABP complexity $O(l)$ such that 
$\po p(x_1, \ldots, x_n)^s$ is equal to
\begin{align*}
 \po v_{\po p}^s(x_1^1,..,x_n^k, \chi_{c_1}(\po p_1^{\m C}(\overline{x}^{\m C})), \chi_{c_2}(\po p_1^{\m C}(\overline{x}^{\m C}), ..,\chi_{c_m}(\po p_l^{\m C})(\overline{x}^{\m C})))
\end{align*}
Where $\po p_1, \ldots, \po p_l$ enumerate all subcircuits of $\po p$ (each is of the form $\po p_G$ for some gate $G$ in the definition of $\po p$), $\po p_i^{\m C}$ is the polynomial $\po p_i$ interpreted in the quotient $\m C$ and $\overline{x}^{\m C}$ is a notation for $x_1^{\m C}, \ldots, x_n^{\m C}$.
We want to prove that there is an ABP process defining $\po v_{\po p}^s$ which also defines in its intermediate steps all $\po v_{\po q}^s$ corresponding to all subcircuits $\po q$ of $\po p$ from the set  $\set{\po p_1, \ldots, \po p_l}$. We do it inductively going from the sources (constants or variables) to the output gate of $\po p$.
Whenever $\po q$ is a variable $x_j$ we can see that $\po v_{\po q}^s$ is equal $x_j^s$ so the claim holds. Whenever $\po q$ is a constant $c$ then $\po v_{\po q}^s$ is just $c^s$ and again, the claim holds.

Essential case of the induction is when $\po q(x_1,\ldots, x_n)$ is created by composing basic operation $\po f(x_1, \ldots, x_r)$ of $\m A$ with shorter polynomials $\po q_1(x_1,\ldots,x_n), \ldots, \po q_r(x_1,\ldots,x_n)$. Note that $\{\po q_1, \ldots, \po q_r\} \subseteq \{\po p_1, \ldots, \po p_l\}$ so we can treat the expressions $\chi_{c_i}(\po q_j(\overline{x})^{\m C})$ as variables when defining $\po v_{\po q}^s$. Using the formula for $\po f(x_1, \ldots, x_r)^s$ we can see that
\begin{align*}
\po q(x_1, \ldots, x_n)^s& = 
\\\sum_{i=1..r}\sum_{j=1..k} \po q_i(\overline x)^j \cdot &\po w_{i,j}(\chi_{c_1}(\po q_1^{\m C}(\overline{x}^{\m C})), \chi_{c_2}(\po q_1^{\m C}(\overline{x}^{\m C})), ..,\chi_{c_m}(\po q_r^{\m C}(\overline{x}^{\m C}))\\
 +&\po w(\chi_{c_1}(\po q_1^{\m C}(\overline{x}^{\m C})), \chi_{c_2}(\po q_2^{\m C}(\overline{x}^{\m C})), ..,\chi_{c_m}(\po q_r^{\m C}(\overline{x}^{\m C})))
\end{align*}
Note that $\po w_{i,j}$ and $\po w$ in the above formula are  degree $r$ polynomials. So they can be written as a linear combination of $(rm)^r$ monomials of degree $r$. So these polynomials have a constant $O(1)$ ABP complexity, as $r$ and $m$ are constants depending on $\m A$ and $\alpha$. So it follows from \cref{abp-mult} (and its short proof) that ABP for $\po v_{\po q}^s$ can be created from ABPs for $(\po v_{\po q_i}^j)_{i \in \{1..r\}, j \in \{1..k\}}$ in $O(1)$ steps. In total the polynomial $\po v_{\po p}^s$ has ABP of size at most $O(l)$, as each vertex of $\po p$ introduces only constant number of steps in the overall ABP-process defining $\po v_{\po p}^s$.  

Now consider an arbitrary program $\progb{\po p}{\iota}{T}$ over $\m A$. Let $\iota'$ be an evaluation of $\iota$ in the quotient $\m C$, so if $a_1, a_2$ are two chosen elements in the image of $\iota$ we translate them to $a_1^{\m C}, a_2^{\m C}$ in $\iota'$. We established that 
$\prog{\po p}{\iota}(b_1, \ldots, b_n)^s$ is equal to
\[ \po v_{\po p}^s(\iota(b_1)^1,..,\iota(b_n)^k, \chi_{c_1}(\po p_1^{\m C}(\iota'(\overline{b}))), \chi_{c_2}(\po p_1^{\m C}(\iota'(\overline{b})), ..,\chi_{c_m}(\po p_l^{\m C}(\iota'(\overline{b})))
\]
Where $\po v_{\po p}^s$ can be represented by ABP of size $O(l)$. Note that\break $\chi_{c}(\po p_j^{\m C}(\iota'(\overline{b}))$ behaves exactly like the evaluation of the program $\progb{\po p_j^{\m C}}{\iota'}{\set{c}}(\overline{b})$ in $\m C$. And each $\iota(b_i)$ can be affinely interpolated in $GF(p)$. So there is a polynomial $\po t_{s}$ over $GF(p)$ of arity $n + kl$ with ABP of size $O(l)$ such that $\prog{\po p}{\iota}(b_1, \ldots, b_n)^s$ is equal to
\begin{equation}\label{final-eq}
\po t_{s}(b_1, \ldots, b_n, (\progb{\po p_j^{\m C}}{\iota'}{\{c\}}(b_1, \ldots, b_n))_{j \in \{1..l\}, c \in C}))
\end{equation}

Now notice that the value of the program $\progb{\po p}{\iota}{T}$ is determined by the set of values $\prog{\po p}{\iota}(b_1, \ldots, b_n)^s$ for $s \in \{1..k\}$ and the  values $\progb{\po p^{\m C}}{\iota'}{\set{t'}}(\overline{b})$ for $t' \in C$. This means there is bounded arity function $\po d: GF(p)^{k+|C|} \longrightarrow GF(p)$ such that
\begin{equation}\label{d-final}
\po d\Big((\prog{\po p}{\iota}{t}(b_1, \ldots, b_n)^s)_{s \in \{1..k\}}, (\progb{\po p^{\m C}}{\iota'}{\set{t'}}(b_1,..,b_n))_{t' \in C}\Big) = 1
\end{equation}
iff $\progb{\po p}{\iota}{T}(\overline{b}) = 1$. Each argument of $\po d$ in the formula above is the composition of the polynomial over $GF(p)$ of ABP complexity $O(l)$ with some number of programs over $\m C$. Indeed we can see from \cref{final-eq} that all arguments are programs, except $b_1, \ldots b_n$, which can be replaced by trivial identity programs over $\m C$. And the programs $\progb{\po p^{\m C}}{\iota'}{\set{t'}}(b_1,..,b_n)$ are itself a program over $\m C$ so we can just compose them on the outside with a trivial identity $x$ polynomial of ABP complexity $1$. In effect we can apply \cref{compose_poly_const} to formula \cref{d-final} and show that function computed by $\progb{\po p}{\iota}{T}$ can be also computed by composition of a polynomial which has ABP representation of size $O(l)$ with bunch of programs over $\m C$. Size of these programs is upper bounded by the size of $\po p$, which finishes the proof 
\end{proof}

Now we are to characterize $\palg{\m A}$ for solvable algebras $\m A$.

\begin{theorem}\label{solv-characterize}
    Let $\m A$ be a finite solvable Malcev algebra. Then there exists sequence of primes $p_1, \ldots, p_h$ such that
    \[
    \palg{\m A} \subseteq \nudet{p_1} \circ \ldots \circ \nudet{p_h}
    \]
\end{theorem}
\begin{proof}
Let $0_A = \alpha_0 \prec \alpha_1 \prec \ldots \prec \alpha_h = 1_A$ be a tight chain of congruences of $\m A$. All the prime quotients $\alpha_{i-1} \prec \alpha_i$ have type $2$, so they have some prime characteristic $p_{i}$ corresponding to their traces.

We can see that $\m C = \m A/\alpha_{h-1}$ is a simple algebra so by \cref{coll:simaff} any function computed by a program over $\m C$ of length $l$ can be also computed by $\MODG{p_{h-1}} \circ \ANDG{d}$ circuit of size $O(\text{poly}(l))$, where $d$ is some fixed constant depending on $\m A$ and $\alpha_{h-1}$. Now we apply the \cref{solvable-induction} to the algebra $\m A/\alpha_{h-1-i}$ and its atom $\alpha_{h-i}$ (for $i \in \{1,..,h-1\}$) to see that function computed by a program over $\m A/\alpha_{i-1-h}$ of size $l$ can be also computed by $\BABPG{p_{h-i}} \circ \PROGG{\m A/\alpha_{h-1-i}}$ circuit of size $O(\text{poly}(l))$. By applying an induction along the chain $0_A = \alpha_0 \prec \alpha_1 \prec \ldots \prec \alpha_h = 1_A$ we can see that any function computed by a program $\po q$ of $\m A$ 
can be also computed by a circuit $C_{\po q}$ of the form 
\[
\BABPG{p_1} \circ \BABPG{p_2} \circ \ldots \BABPG{p_{h-1}} \circ \MODG{p_{h}} \circ \ANDG{d}
\]
Note that every $\MODG{p_{h}} \circ \ANDG{d}$ subcircuit of size $l$ can be replaced by one $\nudet{p_h}$ gate of size $O(\text{poly}(l))$, because $\MODG{p_{h}} \circ \ANDG{d}$ essentially evaluates a bounded degree polynomial over $GF(p)$, and $n$-ary bounded degree $d$ polynomials have ABP complexity at most $n^d$. This finishes the proof. 
\end{proof}

From \cref{solv-nc_2} and from the fact that $NC^2$ is closed under taking fixed height compositions we have

\begin{theorem}\label{theorem-solv-nc2}
Let $\m A$ be a solvable Malcev algebra. Then $\palg{A} \subseteq NC^2$.
\end{theorem}

To fully characterize circuit complexity classes over solvable algebras we now prove the opposite of \cref{solv-characterize}.
\begin{theorem}\label{solv-construc}
For each sequence of primes $p_1, \ldots, p_h$ there exists a finite solvable Malcev algebra $\m A$ such that.
\[
\nudet{p_1} \circ \ldots \circ \nudet{p_h} \subseteq \palg{\m A}
\]
\end{theorem}
\begin{proof}
    We use the the construction used in  \cref{nudet-exact}. Let the underlying set of $A$ be $(Z_{p_1})^2 \times \ldots \times (Z_{p_h})^2$. The idea is that on every pair of coordinates using the same prime $p$ we emulate all the ABPs over $GF(p)$ and then pass them to lower coordinates. So for the prime $p_i$ we introduce operations $(+)^{p_i}, (\pi_1)^{p_i}, (\pi_2)^{p_i}, (u)^{p_i}, (v)^{p_i}$ like in the \cref{nudet-exact}, which act only on coordinates corresponding to $p_i$ and put all the other coordinates in the image to $0$. Then we also introduce unary polynomials $\chi_{i, S}$ for every $i\in \{2,.., h\}$ and $S\subseteq Z_{p_i}$ that look at the left-most coordinate corresponding to $p_i$ and and return $1$ to right-most coordinate corresponding to $p_{i-1}$ iff the argument is in the accepting set $S$, otherwise they return 0. In this way we can build a $\nudet{p_1} \circ \ldots \circ \nudet{p_h}$ circuit by building an appropriate ABPs with $(+)^p, (\pi_1)^p, (\pi_2)^p, u^p, v^p$ and using $\chi_{i, S}$ to simulate accepting sets and for passing values to the coordinates that are more to the left. Instruction $\iota$ should identify $0$ with $(0,\ldots, 0)$ and $1$ with $(0,\ldots, 0, 1)$.
    One can see that in this way we can simulate arbitrary $\BABPG{p_1} \circ \ldots \circ \BABPG{p_h}$ circuit of size $O(l)$ as program of size $O(\text{poly}(l))$  which complete the proof.
\end{proof}

If $C$ is a class of algebras by $\palg{C}$ we mean $\bigcup_{\m A \in C} \palg{\m A}$. What follows is the following.
\begin{corollary}
Let $SOLVABLE$ be the class of all finite solvable
Malcev algebras. Then $\palg{SOLVABLE} = \nuitdet$.
\end{corollary}

It is a common saying in the arithmetic circuit complexity papers that arithmetic formulas are determinants in disguise. Here we see that circuits over solvable algebras are iterated determinants in disguise.

\section{Nilpotent algebras}

In this section we show that circuits over nilpotent algebras recognize exactly languages from the class $CC^0$. This algebras were historically very important in search for tractable cases in the field of equational satisfiability problems for finite algebras. It was suggested already in \cite{kompatscher2022cc} that expressive power of circuits over nilpotent algebras is tied to class $CC^0$. The next Lemma follows from the analysis of the proof in  \cite{IdziakKK25} (Theorem 9.7, Claim 9.3). 
\begin{lemma}\label{lem-nilp}
Let $\m A$ be a finite nilpotent Malcev algebra and $\alpha$ be an atom of characteristic $p$. Then any function computed by a program $\progb{p}{\iota}{S}$ over $\m A$ of length $l$ can  be also computed by $\circuit{\MODG{p}}{\ANDG{d}}{\PROGG{\m A/\alpha}}$ circuit of size $O(\text{poly}(l))$.
\end{lemma}
\noindent For the proof of the above Lemma see the Appendix \ref{apend-lem29}.

\begin{theorem}\label{nilp-theorem}
    Let $\m A$ be a finite nilpotent Malcev algebra. Then $\palg{\m A} \subseteq CC^0$.
\end{theorem}

\begin{proof}
Take an algebra $\m A$ and a tight chain of congruences $0_A = \alpha_0 \prec \alpha_1 \prec \ldots \prec \alpha_h = 1_A$. Define a quotient algebra $\m A_j = \m A/\alpha_{h-j}$. By \cref{lem-nilp} we can see that a program of $\m A_j$ of size $l$ can be represented by $\text{poly}(l)$-size $\circuit{\MODG{p_j}}{\ANDG{d}}{\PROGG{\m A_{j-1}}}$ circuit, where $p_j$ is a characteristic of $\alpha_{j-1} \prec \alpha_{j}$ prime quotient. Thus by induction on $j$ we can see that $\progb{\po p}{\iota}{S}$ itself can be computed by $\text{poly}(l)$ size circuit of type 
\[
    \MODG{p_1} \circ \ANDG{d} \circ \MODG{p_2} \circ \ANDG{d} \circ \ldots \circ \MODG{p_{h}} 
\]

We can replace each $\ANDG{d}$ gate by $\MODG{2} \circ \MODG{3}$ circuit of constant size (as every Boolean function has such $\MODG{2} \circ \MODG{3}$ circuit \cite{BarringtonST90}), and it finishes the proof of theorem (as we can simulate each $\MODG{p}$ gate by $\MODG{m}$ gate for $m$ being the lcm of primes involved).
\end{proof}

The reverse is also true and $CC_h[m]$ circuits can be simulated by some nilpotent Malcev algebra.
\begin{example}\label{example-nilpo}
Let $\m A$ be such that $A = (\mathbb{Z}_m)^h$ and $\m A$ has the following operations.
\begin{enumerate}
    \item abelian binary $+$ taken from $(\mathbb{Z}_m,+)^h$,
    \item for each $i \in \{1,\ldots,h\}$ a unary projection $\po e_i$ such that $(\po e_i(x))_j = x_i$ iff $i=j$ and  $(\po e_i(x))_j = 0$ otherwise,
    \item for each $i \in \{2,\ldots,h\}$ and $S\subseteq \mathbb{Z}_m$ unary $\chi_{S,i}$ such that $\chi_{S,i}(x)_j =1$ iff $j=i-1 \land x_i \in S$  and $(\chi_{S,i}(x))_j = 0$ otherwise.
\end{enumerate}
\end{example}

It is easy to check that one can rewrite any $CC_h[m]$ circuit into a program of this algebra in an efficient way. Indeed just replace the gate $\MODG{m}$ with an accepting set $S$ on the $i$-th level of the circuit using projection $\po e_i$ and $+$ operation and then pass the resulting value with $\chi_{S,i}$ to the next level. This algebra is Malcev because it has an abelian $+$ and is nilpotent because congruences forgetting prefixes of coordinates form a proper centralizing series.
\begin{corollary}
Let $NILPOTENT$ be class of all finite nilpotent Malcev algebras. Then $\palg{NILPOTENT} = CC^0$.
\end{corollary}
\section{Products of solvable and totally ordered algebras}

Now we are ready to prove the two last points of \cref{thm-main}. 

\begin{theorem}\label{hard-combination-types}
Let $\m A$ be a finite algebra from a congruence modular variety. If there exist join irreducible congruences $\alpha,\beta\in\con{A}$ such that $\beta\leq\alpha$,  $(\alpha^-,\alpha)$-minimal sets are of type $2$ and $(\beta^-,\beta)$-minimal sets are of type $4$, then $\palg{\m A}=\ppoly$.
\end{theorem}
\begin{proof}
Let $U$ be $(\alpha^-,\alpha)$-minimal, and $V$ be $(\beta^-,\beta)$-minimal set. Let $e_U, e_V$ be some projections to the respective minimal set. Pick arbitrary two distinct elements in $U$ which are equal modulo $\alpha$ but not modulo $\alpha^-$, call them $0_U, 1_U$, and let $0_V, 1_V$ be the two elements of $V$. From \cite[Lemma~4.20]{hm}
there is a ternary operation $\po d(x,y,z)$ such that $\po d(y,y,x) = \po d(x,y,y) = x$ for all $x,y \in U$.

First we claim that there exists unary polynomial $\po f$ in $\pol{A}$ such that $\po f(0_U) = (0_V)$ and $\po f(1_U) = 1_V$. Since $\alpha$ is join irreducible the pair $(0_U, 1_U)$ closed under unary polynomials of $\m A$ must connect all the pairs in congruence $\alpha$, in particular $0_V, 1_V$.  Hence the undirected graph which vertices are elements of $\m A$ and edges are of the form $\set{\po g(0_U), \po \po g(1_U)}$ for all unary polynomials $\po g$, must connect $0_V$ with $1_V$. But since $\po e_V(0_V) = 0_V$ and $\po e_V(1_V) = 1_V$, if we replace each edge $\set{\po g(0_U), \po g(1_U)}$ by $\set{\po e_V(\po g(0_U)), \po e_V(\po g(1_U))}$ the resulting graph still connects $0_V$ with $1_V$. But the image of $\po e_V$ is equal to $\{0_V, 1_V\}$, hence there must be some unary $\po f$ such that $\po f(\set{0_U, 1_U}) = \set{0_V, 1_V}$. Then,  $\po f$ or $\po f( \po d(0_U, x, 1_U))$ does the job.

Now take arbitrary Boolean circuit $C(b_1, \ldots, b_n)$. We may assume by the Morgans' laws that negations are only in leafs and all the other gates are $\lor, \land$. To simulate circuit $C$ by the program of $\m A$ pick $\iota(0) = 0_U, \iota(1) = 1_U$ and the accepting set $S = \{1_V\}$. Circuit $\po p$ over $\m A$ is created by replacing $\land, \lor$ in $C$ by $\land_{V}, \lor_V$ and whenever variable $b_i$ is nonnegated we just replace it by $\po f(x_i)$ and when its negated we use $\po f(\po d(0_U, x_i, 1_U))$ instead. One can see that the program $\progb{\po p}{\iota}{S}$ has linear size in terms of the size of $C$ and computes the exact same function.

\end{proof}

To prove the last point we need the following algebraic lemma, which allows us to separate lattice behaviour from the affine one.

\begin{lemma}\label{last-lem}
Let $\m A$ be a finite algebra from a congruence modular variety such that $\typset{\m A} \subseteq \{2,4\}$. If there are no join irreducible congruences $\alpha,\beta\in\con{A}$ such that $\beta\leq\alpha$,  $(\alpha^-,\alpha)$-minimal sets are of type $2$ and $(\beta^-,\beta)$-minimal sets are of type $4$, then there exists a congruence $\theta \in \con{A}$ such that $\typset{\m A/\theta} = \{4\}$ and all prime quotients below $\theta$ are of type $2$.
\end{lemma}
\begin{proof}
Valeriote in \cite{valeriote1986decidable} introduced so called transfer principle. We say that a finite algebra satisfies the $(i, j)$-transfer principle if whenever there are prime quotients $\alpha\prec\beta$ of type $i$ and $\beta\prec\gamma$ of type $j$  there exists $\beta'$ such that $\alpha\prec\beta'$ is a prime quotient of type $j$ and $\beta'\leq\gamma$. If $\m A$ satisfies $(4,2)$-transfer principle then we can transform any path from $0_{\m A}$ to $1_{\m A}$ in congruence lattice of $\m A$ in such a way that it is of the form $0_{\m A}=\alpha_1\prec\alpha_2\ldots\alpha_{h-1}\prec\alpha_{h}=1_{\m A}$ and there exists $k$ such that $\alpha_i\prec\alpha_{i+1}$ is of type $2$ for $i<k$ and of type $\tn 4$ for $i\geq k$.  In such a case $\theta=\alpha_k$ do the job. If $\m A$ does not satisfy $(4,2)$-transfer principle then by \cite[Lemma 6.2]{IdziakK22} there are congruences $\alpha\prec\beta\prec\gamma$ such that $\alpha\prec\beta$ is a quotient of type $\tn 4$, $\beta\prec\gamma$ is a quotient of type $\tn 2$ and $\gamma$ is join irreducible. Contradiction, since $\alpha\prec\beta$ can be transposes down to quotient $\alpha'\prec\beta'$ of type $\tn 4$ with $\beta'$ being join irreducible.
\end{proof}

This leads us to the following.
\begin{theorem}\label{final-mixed-th}
Let $\m A$ be a finite algebra from a congruence modular variety such that $\typset{\m A} \subseteq \{2,4\}$ and there exists a congruence $\beta \in \con{A}$ such that $\typset{\m A/\beta} = \{4\}$ and all prime quotients below $\beta$ are of type $2$. If $\m A/\beta$ is a subdirect product of totally ordered algebras
\[
\palg{\m A} \subseteq \circuittwo{\nuitdet}{\MON}\ \ (\subseteq NC^2 \circ \MON)
\]
\end{theorem}
\begin{proof}
    We can see that if $\beta = 1_{\m A}$ the statement follows from \cref{theorem-solv-nc2}. If $\beta$ is nontrivial, we can choose a sequence of congruences $0_A = \alpha_0 \prec \alpha_1 \prec \ldots \prec \alpha_h = \beta$. Each prime quotient $\alpha_{i-1} \prec \alpha_i$ is of type $2$ and has some prime characteristic $p_i$. Therefore by \cref{lm-subdirect-type-4} and \cref{solvable-induction} iteratively applied along a chain $\alpha_i$ we can see that 
    \[
    \palg{A} \subseteq \nudet{p_1} \circ \ldots \circ \nudet{p_h} \circ \MMON
    \]
    But now any circuit of the class $\MMON$ is a composition of bounded arity function with some number of monotone circuits. Any bounded arity Boolean function is computable with the polynomial of constant ABP complexity over any finite field with accepting set $\{1\}$ so 
    \[
    \palg{A} \subseteq \nuitdet \circ \MON
    \]
\end{proof}

\section{Simple algebras}\label{sec-simple}
In this section we will consider simple algebras i.e. algebras witch have no nontrivial congruences. In such a case $0\prec 1$ is an only prime quotient and hence every two minimal sets are polynomially equivalent. Hence, it has sense to call a type of minimal sets a type of algebra.  Moreover, every minimal set has exactly one trace and no tail.  Most of our results can be summarized in the following theorem.
\begin{theorem}\label{thm-simple}
Let $\m A$ be simple finite algebra then if its type is equal:
\begin{enumerate}
    \item[$\tn 1$] then $\palg{\m A}=\BOUNDA{c}$, for some constant $c$.
    \item[$\tn 2$] then $\palg{\m A}=\BOUNDA{c} \circ \MODG{p}$, for some constant $c$ and prime number $p$.
    \item[$\tn 3$] then $\palg{\m A}=\ppoly$.
    \item[$\tn 4$] then $\MON \subseteq \palg{\m A} \subseteq \MMON$ for totally ordered $\m A$ and $\palg{A} = \ppoly$ otherwise.
\end{enumerate}
\end{theorem}
Before we will prove \cref{thm-simple} let us note that simple algebras of type $\tn 5$ are not mentioned. It is because of surprisingly rich structure of these algebras and the resulting problem in simply characterizing the languages they recognize. It turns out that for many of considered languages there is a finite simple algebra  of type $5$ recognizing languages from this class. The following example can give us a taste of what might happen in case of algebras of type $\tn 5$.

\begin{example}\label{ex-type-5}
Let $A=\set{\bot,0,1}$ $\po e_1$, $\po e_2$, $\po e_3$, $\neg$, $\wedge$, $\vee$, $+$ be operations on $A$ such that
\begin{itemize}
    \item $\po e_1(\bot)=\po e_1(0)=\bot$ and $\po e_1(1)=0$,
    \item $\po e_2(\bot)=\bot$ and $\po e_2(0)=\po e_2(1)=1$,
     \item $\po e_3(\bot)=\bot$ and $\po e_3(0)=\po e_3(1)=0$,
    \item $\neg \bot = \bot$, $\neg 0=1$ and $\neg 1=0$,
    \item $\bot \wedge x=\bot \wedge x=\bot$ for $x\in A$, $0 \wedge y=y \wedge 0=0$ for $y\in\set{0,1}$ and $1\wedge 1=1$,
    \item $\bot \vee x=\bot \vee x=\bot$ for $x\in A$, $0\vee 0=0$ and $1' \vee y=y \vee 1 =1$ for $y\in\set{0,1}$,
     \item $\bot + x=\bot + x=\bot$ for $x\in A$, $0 + 1=1+0=1$, $0 + 0=1+1=0$.
\end{itemize}
Then for algebras
\begin{itemize}
    \item $\m A_1=(a,\po e_{1},\po e_{2}, \po e_3, \wedge)$,
    \item $\m A_2=(a,\po e_{1},\po e_{2}, \po e_3, \wedge,\vee)$,
    \item $\m A_3=(a,\po e_1,\po e_{2}, \po e_3, \wedge,\neg)$,
    \item $\m A_4=(a,\po e_1,\po e_2, \po e_3, +)$,
\end{itemize}
it holds that
\begin{itemize}
    \item $\MON\nsubseteq\palg{\m A_1}\subseteq\MMON$,
    \item $\MON\subseteq\palg{\m A_2}\subseteq\MMON$,
    \item $\palg{\m A_3}=\ppoly$,
    \item $P_{\m A_4}\subseteq \mathsf{AC}_k[2]$ for some small constant $k$.
\end{itemize}
\end{example}

From here until the end of the section we will prove \cref{thm-simple}. It is known that for every finite simple algebra $\m A$ of type 1 (unary type) there is a uniform bound $k_{\m A}$ for the essential arity of its polynomial operations. What follows, also programs over $\m A$ can express only essentialy-constant arity operations. As already the expressive power of programs over such algebras is narrow, studying their corresponding circuit complexity is pointless.

Similarly as in the previous sections we will use representations of finite algebras that are inspired by the \cite[Chapter\ 13]{hm}. In particular Lemma 13.1 provides representations for type 1 and type 2 simple algebras. 

\begin{lemma}\label{lem:un_aff}
Let $\m N$ be a trace of a simple algebra $\m A$ of type 1 (unary type) or type 2 (affine type). Then $\m A$ is a subreduct of $\m N^{[k]}$ for some integer $k$.
\end{lemma}
What follows is the following
\begin{observation}
    Let $\m A$ be a finite simple algebra of type $1$ (unary type). Then there is a constant $k_{\m A}$ such that for every $f \in \pol{A}$ the essential arity of $f$ is bounded by $k_{\m A}$.
\end{observation}
\begin{proof}
Without loss of generality we assume that $\m A$ is a reduct of $\m N^{[k]}$.  We will prove that $k_{\m A} := k$ does the job. 

From the definition of matrix power we have that 
\[
f(x_1, \ldots, x_n) = (f_1(\ov{x}), \ldots, f_n(\ov{x}))
\]
where each $f_i$ is a polynomial of $\m N$ over variables $(x_1)_1, \ldots, (x_n)_k$. But every polynomial $f_i$ (as a polynomial of $\m N$) depends only on one variable. Thus $f$ depends on at most $k$ variables, each one corresponding to the variable on which $f_i$ depends.
\end{proof}
\begin{corollary}\label{coll:simun}
Let $\m A$ be a finite simple algebra of type 1 (unary type). Let $k_{\m A}$  be a universal bound on the essential arity of polynomials of $\m A$. Them $k_{\m A}$ bounds the essential arity of programs over $\m A$.
\end{corollary}
\begin{proof}
    If polynomial $\po p(\overline{x})$ of $\m A$ depends only on variables $x_{i_1}, \ldots, x_{i_k}$ then any program $\progb{p}{\iota}{S}(\overline{b})$ depends only on variables $b_{i_1}, \ldots, b_{i_k}$.
\end{proof}

The first point of \cref{thm-simple} is a straightforward consequence of \cref{coll:simun}.

The affine type is similar to the unary type in the sense that it also provides a very narrow expressive power when it comes to polynomials/programs. It follows from \cref{coll:simun} that every simple abelian algebra is a subreduct of of matrix power of a trace of type $\tn 2$ and hence is subreduct of module. Then, by \cref{coll:simaff} we have that NUDFA's over simple algebra $\m A$ of  affine type recognize languages contained in $\palg{\m A}=\BOUNDA{c} \circ \MODG{p}$ for come constant $c$ and prime $p$.

NuDFA's over type $3$ simple algebras can recognize exactly languages from $\ppoly$ as we can choose $\iota$ which maps $\set{0,1}$ onto minimal set of type $\tn 3$ which is polynomially equivalent with Boolean algebra. Characterization for simple algebras of type $4$ is an easy consequence of \cref{lm-subdirect-type-4}.

\section{Final remarks and open problems}
We explored complexity classes $\palg{\m A}$ in large, but still restricted class of structures. It is not clear how to characterize computational classes induced by simple algebras of type $5$. It is clearly visible that just fundamental tools of Tame Congruence Theory are not enough for the analysis. Understanding this case better could open the door to study general finite algebraic structures in more depth, possibly allowing for some strong representation theorems. 

It is also quite easy to see that using type $5$ allows us to capture more of the standard classes such as $AC^0, AC[p], ACC^0$ studied in the computational complexity theory. Indeed one can use semilattices as layers, similarly to the proof of \cref{solv-construc} or \cref{example-nilpo}, and possibly mix them with some addition modulo layers (depending whether we want to achieve $AC^0$, $AC^0[p]$ or $ACC$). Then projections from the upper levels to the lower levels would allow for negation, which enables us to define $AC^0, AC[p], ACC^0$ as limits of such constructions. We can see that $\palg{C}$ for some easy to define classes/varieties of algebras $C$ can capture truly a lot of fundamental computational classes.

It is quite interesting to study whether some other natural classes like $TC^0$ can be captured by some natural class of algebras. One can not only  use the mysterious behaviours of type $5$ algebras, but also try to define some solvable subvarieties with the desired properties. 

\begin{open}
    Does there exist a class $C$ of finite algebras such that $\palg{C} = TC^0$?
\end{open}

In view of \cref{thm-dichotomy} the following open problem seems to be important, however probably out of reach for the current techniques.
\begin{open}\label{op-mm*nc2}
Is $\circuittwo{\mathsf{NC}^2}{\MON} \subsetneq \mathsf{P}/poly$?
\end{open}

\printbibliography

@inproceedings{dvir2012separating,
  title={Separating multilinear branching programs and formulas},
  author={Dvir, Zeev and Malod, Guillaume and Perifel, Sylvain and Yehudayoff, Amir},
  booktitle={Proceedings of the forty-fourth annual ACM symposium on Theory of computing},
  pages={615--624},
  year={2012}
}

@InProceedings{IdziakKK25,
  author =	{Idziak, Pawe{\l} M. and Kawa{\l}ek, Piotr and Krzaczkowski, Jacek},
  title =	{{Nonuniform Deterministic Finite Automata over Finite Algebraic Structures}},
  booktitle =	{52nd International Colloquium on Automata, Languages, and Programming (ICALP 2025)},
  pages =	{161:1--161:14},
  series =	{Leibniz International Proceedings in Informatics (LIPIcs)},
  ISBN =	{978-3-95977-372-0},
  ISSN =	{1868-8969},
  year =	{2025},
  volume =	{334},
  publisher =	{Schloss Dagstuhl -- Leibniz-Zentrum f{\"u}r Informatik},
  address =	{Dagstuhl, Germany},
  URN =		{urn:nbn:de:0030-drops-235386},
  doi =		{10.4230/LIPIcs.ICALP.2025.161}
}

@inproceedings{Barrington86,
  author    = {David A. Mix Barrington},
  title     = {Bounded-Width Polynomial-Size Branching Programs Recognize Exactly
               Those Languages in {NC}{\({^1}\)}},
  booktitle = {Proceedings of STOC'86},
  pages     = {1--5},
  doi       = {10.1145/12130.12131},
  bibsource = {dblp computer science bibliography, https://dblp.org},
  year      = {1986}
}

@InProceedings{IdziakKKW22-icalp,
  author =	{Idziak, Pawe{\l} M. and Kawa{\l}ek, Piotr and Krzaczkowski, Jacek and Wei{\ss}, Armin},
  title =	{{Satisfiability Problems for Finite Groups}},
  booktitle =	{49th International Colloquium on Automata, Languages, and Programming (ICALP 2022)},
  pages =	{127:1--127:20},
  series =	{Leibniz International Proceedings in Informatics (LIPIcs)},
  ISBN =	{978-3-95977-235-8},
  ISSN =	{1868-8969},
  year =	{2022},
  volume =	{229},
  publisher =	{Schloss Dagstuhl -- Leibniz-Zentrum f{\"u}r Informatik},
  address =	{Dagstuhl, Germany},
  doi =		{10.4230/LIPIcs.ICALP.2022.127},
}

@article{BarringtonST90,
	author = {David A. Mix Barrington and Howard Straubing and Denis Th{\'e}rien},
	bibsource = {dblp computer science bibliography, https://dblp.org},
	doi = {10.1016/0890-5401(90)90007-5},
	journal = {Inf. Comput.},
	number = {2},
	pages = {109--132},
	title = {Non-Uniform Automata Over Groups},
	url = {https://doi.org/10.1016/0890-5401(90)90007-5},
	volume = {89},
	year = {1990}
}

@article{GoldmannR02,
	author = {Mikael Goldmann and Alexander Russell},
	bibsource = {dblp computer science bibliography, https://dblp.org},
	doi = {10.1006/inco.2002.3173},
	journal = {Inf. Comput.},
	number = {1},
	pages = {253--262},
	title = {The Complexity of Solving Equations over Finite Groups},
	url = {https://doi.org/10.1006/inco.2002.3173},
	volume = {178},
	year = {2002}
}

@online{BurrisS12,
  author    = {Burris, Stanley N. and Sankappanavar, H. P.},
  title     = {A Course in Universal Algebra},
  year      = {2012},
  url       = {https://www.math.uwaterloo.ca/~snburris/htdocs/ualg/univ-algebra2012.pdf},
  note      = {Millennium Edition, freely available PDF},
  urldate   = {2026-01-17}
}

@article{kompatscher2022cc,
  title={{CC}-circuits and the expressive power of nilpotent algebras},
  author={Kompatscher, Michael},
  journal={Logical Methods in Computer Science},
  volume={18},
  year={2022},
  publisher={Episciences. org}
}

@article{IdziakK22,
  title={Satisfiability in multivalued circuits},
  author={Idziak, Pawe{\l} M and Krzaczkowski, Jacek},
  journal={SIAM Journal on Computing},
  volume={51},
  number={3},
  pages={337--378},
  year={2022},
  publisher={SIAM}
}

@book{hm,
  title={Structure of Finite Algebras},
  author={David Hobby and Ralph McKenzie},
  year={1988},
  publisher={Contemporary Mathematics vol. 76. American Mathematical Society},
  doi={10.1090/conm/076}
}

@book{fm,
  title={Commutator Theory for Congruence Modular Varieties},
  author={Freese, Ralph and McKenzie, Ralph},
  year={1987},
  publisher={London Mathematical Society Lecture Notes, No. 125. Cambridge University Press}
}

@inproceedings{valiant1979completeness,
  title={Completeness classes in algebra},
  author={Valiant, Leslie G},
  booktitle={Proceedings of the eleventh annual ACM symposium on Theory of computing},
  pages={249--261},
  year={1979}
}

@article{malod2008characterizing,
  title={Characterizing Valiant's algebraic complexity classes},
  author={Malod, Guillaume and Portier, Natacha},
  journal={Journal of complexity},
  volume={24},
  number={1},
  pages={16--38},
  year={2008},
  publisher={Elsevier}
}

@article{berkowitz1984computing,
  title={On computing the determinant in small parallel time using a small number of processors},
  author={Berkowitz, Stuart J},
  journal={Information processing letters},
  volume={18},
  number={3},
  pages={147--150},
  year={1984},
  publisher={Elsevier}
}

@article{samuelson1942method,
  title={A method of determining explicitly the coefficients of the characteristic equation},
  author={Samuelson, Paul A},
  journal={The Annals of Mathematical Statistics},
  volume={13},
  number={4},
  pages={424--429},
  year={1942},
  publisher={JSTOR}
}

@article{mahajan1997determinant,
  title={Determinant: Combinatorics, algorithms, and complexity},
  author={Mahajan, Meena and others},
  journal={Chicago Journal of Theoretical Computer Science},
  volume={1997},
  number={5},
  year={1997},
  publisher={AI2}
}

@article{lund1992algebraic,
  title={Algebraic methods for interactive proof systems},
  author={Lund, Carsten and Fortnow, Lance and Karloff, Howard and Nisan, Noam},
  journal={Journal of the ACM (JACM)},
  volume={39},
  number={4},
  pages={859--868},
  year={1992},
  publisher={ACM New York, NY, USA}
}

@article{shamir1992ip,
  title={{IP}= {PSPACE}},
  author={Shamir, Adi},
  journal={Journal of the ACM (JACM)},
  volume={39},
  number={4},
  pages={869--877},
  year={1992},
  publisher={ACM New York, NY, USA}
}

@inproceedings{sankaranarayanan2004non,
  title={Non-linear loop invariant generation using Gr{\"o}bner bases},
  author={Sankaranarayanan, Sriram and Sipma, Henny B and Manna, Zohar},
  booktitle={Proceedings of the 31st ACM SIGPLAN-SIGACT symposium on Principles of programming languages},
  pages={318--329},
  year={2004}
}

@inproceedings{rodriguez2004automatic,
  title={Automatic generation of polynomial loop invariants: Algebraic foundations},
  author={Rodr{\'\i}guez-Carbonell, Enric and Kapur, Deepak},
  booktitle={Proceedings of the 2004 international symposium on Symbolic and algebraic computation},
  pages={266--273},
  year={2004}
}

@article{lecun2015deep,
  title={Deep learning},
  author={LeCun, Yann and Bengio, Yoshua and Hinton, Geoffrey},
  journal={nature},
  volume={521},
  number={7553},
  pages={436--444},
  year={2015},
  publisher={Nature Publishing Group UK London}
}

@article{vaswani2017attention,
  title={Attention is all you need},
  author={Vaswani, Ashish and Shazeer, Noam and Parmar, Niki and Uszkoreit, Jakob and Jones, Llion and Gomez, Aidan N and Kaiser, {\L}ukasz and Polosukhin, Illia},
  journal={Advances in neural information processing systems},
  volume={30},
  year={2017}
}

@inproceedings{bulatov2017dichotomy,
  title={A dichotomy theorem for nonuniform {CSP}s},
  author={Bulatov, Andrei A},
  booktitle={2017 IEEE 58th Annual Symposium on Foundations of Computer Science (FOCS)},
  pages={319--330},
  year={2017},
  organization={IEEE}
}

@article{zhuk2020proof,
  title={A proof of the CSP dichotomy conjecture},
  author={Zhuk, Dmitriy},
  journal={Journal of the ACM (JACM)},
  volume={67},
  number={5},
  pages={1--78},
  year={2020},
  publisher={ACM New York, NY, USA}
}

@article{barrington1990non,
  title={Non-uniform automata over groups},
  author={Barrington, David A Mix and Straubing, Howard and Th{\'e}rien, Denis},
  journal={Information and Computation},
  volume={89},
  number={2},
  pages={109--132},
  year={1990},
  publisher={Elsevier}
}

@article{barrington1988finite,
  title={Finite monoids and the fine structure of {NC}},
  author={Barrington, David A Mix and Therien, Denis},
  journal={Journal of the ACM (JACM)},
  volume={35},
  number={4},
  pages={941--952},
  year={1988},
  publisher={ACM New York, NY, USA}
}

@inproceedings{barrington2000equation,
  title={Equation satisfiability and program satisfiability for finite monoids},
  author={Barrington, David Mix and McKenzie, Pierre and Moore, Cris and Tesson, Pascal and Th{\'e}rien, Denis},
  booktitle={International Symposium on Mathematical Foundations of Computer Science},
  pages={172--181},
  year={2000},
  organization={Springer}
}

@inproceedings{razborov1994natural,
  title={Natural proofs},
  author={Razborov, Alexander A and Rudich, Steven},
  booktitle={Proceedings of the twenty-sixth annual ACM symposium on Theory of computing},
  pages={204--213},
  year={1994}
}

@article{pudlak1997lower,
  title={Lower bounds for resolution and cutting plane proofs and monotone computations},
  author={Pudl{\'a}k, Pavel},
  journal={The Journal of Symbolic Logic},
  volume={62},
  number={3},
  pages={981--998},
  year={1997},
  publisher={Cambridge University Press}
}

@article{krajivcek1994lower,
  title={Lower bounds to the size of constant-depth propositional proofs},
  author={Kraj{\'\i}{\v{c}}ek, Jan},
  journal={The Journal of Symbolic Logic},
  volume={59},
  number={1},
  pages={73--86},
  year={1994},
  publisher={Cambridge University Press}
}

@article{razborov1995unprovability,
  title={Unprovability of lower bounds on circuit size in certain fragments of bounded arithmetic},
  author={Razborov, Alexander A},
  journal={Izvestiya: mathematics},
  volume={59},
  number={1},
  pages={205},
  year={1995},
  publisher={IOP Publishing}
}

@inproceedings{bonet1995lower,
  title={Lower bounds for cutting planes proofs with small coefficients},
  author={Bonet, Maria and Pitassi, Toniann and Raz, Ran},
  booktitle={Proceedings of the twenty-seventh annual ACM symposium on Theory of computing},
  pages={575--584},
  year={1995}
}

@article{Razborov1985,
  author  = {A. A. Razborov},
  title   = {Lower bounds on the monotone complexity of some boolean functions},
  journal = {Dokl. Akad. Nauk SSSR},
  volume  = {281},
  number  = {4},
  pages   = {798--801},
  year    = {1985},
  note    = {Translation in Soviet Math. Dokl. 31:354--357}
}

@article{alon1987monotone,
  title={The monotone circuit complexity of Boolean functions},
  author={Alon, Noga and Boppana, Ravi B},
  journal={Combinatorica},
  volume={7},
  number={1},
  pages={1--22},
  year={1987},
  publisher={Springer}
}

@book{Post1941,
  author    = {Post, E. L.},
  title     = {The two-valued iterative systems of mathematical logic},
  series    = {Annals of Mathematics studies},
  number    = {5},
  publisher = {Princeton University Press},
  address   = {Princeton},
  year      = {1941},
  note      = {122 pp.}
}

@article{tardos1988gap,
  title={The gap between monotone and non-monotone circuit complexity is exponential},
  author={Tardos, {\'E}va},
  journal={Combinatorica},
  volume={8},
  number={1},
  pages={141--142},
  year={1988},
  publisher={Springer}
}

@article{cavalar2025monotone,
  title={Monotone circuit complexity of matching},
  author={Cavalar, Bruno and G{\"o}{\"o}s, Mika and Riazanov, Artur and Sofronova, Anastasia and Sokolov, Dmitry},
  journal={arXiv preprint arXiv:2507.16105},
  year={2025}
}

@article{amano2005superpolynomial,
  title={A superpolynomial lower bound for a circuit computing the clique function with at most (1/6) log log n negation gates},
  author={Amano, Kazuyuki and Maruoka, Akira},
  journal={SIAM Journal on Computing},
  volume={35},
  number={1},
  pages={201--216},
  year={2005},
  publisher={SIAM}
}

@phdthesis{valeriote1986decidable,
  author       = {Matthew A. Valeriote},
  title        = {On Decidable Locally Finite Varieties},
  school       = {University of California, Berkeley},
  year         = {1986},
  type         = {Ph.D. thesis},
  address      = {Berkeley, California, USA},
  note         = {Advisor information not confirmed},
}

@incollection{AgostonDH86,
title = {ON THE NUMBER OF CLONES CONTAINING ALL CONSTANTS (A PROBLEM OF {R}. {M}C{K}ENZIE)},
booktitle = {Lectures in Universal Algebra},
publisher = {North-Holland},
address = {Amsterdam},
pages = {21-25},
year = {1986},
series = {Colloquia Mathematica Societatis Janos Bolyai},
isbn = {978-0-444-87759-8},
doi = {https://doi.org/10.1016/B978-0-444-87759-8.50006-7},
author = {Ágoston, I. and Demetrovics, J. and Hannák, L.}
,
}

%\appendix

\newpage
\appendix

\section{Proof of \cref{struct-alp-cong}}

\begin{proof}
We will prove the lemma in two steps. First we will show that $\m A$ is a subreduct of $\m A|_N^{[k_1]}\boxtimes \m A/\alpha$, where $N$ is a trace of the minimal set $U$. Then we will show that $N$ is a subreduct of matrix power $\mathbb{Z}_p^{[k_2]}$.

Let $N_1,\ldots,N_l$ be the traces of $U$ and  $a_i\in N_i$ for each $i$. By \cite[Lemma~4.20]{hm} there exists a pseudo-Malcev term $\po d$ that behaves like a Malcev term on $U$ and such that $\po d(x,a,b)$ is a permutation of $U$ for all $a,b\in U$. Hence the polynomial $\po f_i(x)=\po d(x,a_i,a_1)$ is a permutation of $U$ mapping $N_i$ onto $N_1$, and $\po f_i^{-1}$ is also a polynomial of $\m A|_U$.

Let $\po f_1,\ldots,\po f_{k_1}$ be all unary polynomials of $\m A$ with range contained in $U$. By \cite[Theorem~2.8(4)]{hm}, for $(a,b)\in\alpha$ we have
\[
a=b \iff (\po f_1(a),\ldots,\po f_{k_1}(a))=(\po f_1(b),\ldots,\po f_{k_1}(b)).
\]
Fix $a\in A$ and define $c_a:\{1,\ldots,k_1\}\to\{1,\ldots,l\}$ such that  $c_a(i)=j$ whenever $e_i(a)\in N_j$. Then the mapping
\[
b\mapsto (\po f_{c_a(1)}(e_1(b)),\ldots,\po f_{c_a(k_1)}(e_{k_1}(b)))
\]
is an injection from $a/\alpha$ into $N_1^{k_1}$. Repeating the arguments from the proof of \cref{lm-decomposition} for the abelian case, we obtain
\[
\m A \cong \m A|_{N_1}^{[k_1]}\boxtimes \m A/\alpha.
\]

Finally, $\m A|_{N_1}$ is polynomially equivalent to a one-dimensional vector space over $GF(p^m)$. Hence every polynomial function is of the form $\sum_i s_i(x_i)+c$, where each $s_i$ is an endomorphism of $(\mathbb{Z}_p,+)^m$. Assume without loss of generality that $(N_1,+)$ is just $(\mathbb{Z}_p,+)^m$ and $\pi_i$ is the projection on the $i$-th coordinate. Then, $\pi_i(s_j(x_j))$, for every $i$, is a linear function over $\mathbb{Z}_p$ in variables $\pi_1(x_j), \ldots, \pi_1(x_j)$.Hence, $\m A|_{N_1}$ is a reduct of $\mathbb{Z}_p^{[m]}$.
\end{proof}

\section{Proof of \cref{coll:simaff}}

We need the following
\begin{observation}
Let $\m A$ be an affine simple algebra, pollynomialy equivalent to a module $\m M$ over an abelian group $(\mathbb{Z}_p)^k$ for $p$ being a prime. Then for any $c \in M$ and any polynomial $\po p(\overline{x}) \in \pol{A}$ the equation $\po p(x_1, \ldots, x_n) = c$ is equivalent to system of $k$ equation
$L_i\Big((x_{1})_1, \ldots, (x_{n})_k\Big) = c_i$ for $1 \leq i \leq n$, where each $L_i$ is a linear function in $\mathbb{Z}_p$ and $c_i \in \mathbb{Z}_p$.
\end{observation}

Indeed, Every polynomial $\po p(x_1, \ldots, x_n)$ of a module $\m M$ is of the form $\alpha_1 x_1 + \ldots + \alpha_n x_n + d$ where each $\alpha_i$ is an endomorphisms of an abelian group $(\mathbb{Z}_p)^k$ and $d \in (Z_p)^k$. So after projecting $\alpha_i x_i$ to i'th coordinate ($1 \leq i \leq k$) we get a linear function in variables $(x_i)_1, \ldots, (x_i)_k$. Hence $\alpha_1 x_1 + \ldots + \alpha_n x_n + d = c$ after projecting to $i$'th coordinate can be rewritten to a form $L_i\Big((x_{1})_1, \ldots, (x_{n})_k\Big) = c_i$, as projection itself is a linear map.

\begin{proof}[Proof of \cref{coll:simaff}]
    The condition $\progb{\po p}{\iota}{S}(\overline{b}) = 1$ is equivalent to the bounded arity alternative of conditions $\prog{\po p}{\iota}{(\overline{b})} = s$ over $s\in S$. Each such condition is equivalent to conjunction of conditions 
    \[
    L_i\Big((\iota_1(b_1)_1, \ldots, (\iota_n(b_n)_k\Big) = c^{(s)}_i
    \]
    for some constant $c^{(s)}_i \in \mathbb{Z}_p$.  As $\iota$ is the function from $\{0,1\} \mapsto \mathbb{Z}_p$ it can be replaced (interpolated) with an affine formula over $\mathbb{Z}_p$.  So $L_i\Big((\iota_1(b_1)_1, \ldots, (\iota_n(b_n)_k\Big) = c^{(s)}_i$ is equivalent to 
    \[
    L_i'\Big((b_1)_1, \ldots, ((b_n)_k\Big) = d^{(s)}_i
    \]  for some other linear expression $L_i'$ and $d^{(s)}_i \in \mathbb{Z}_p$. Such a linear condition can be asserted with a single $\MODG{p}^{d^{(s)}_i}$ gate. Hence function computed by a program $\progb{\po p}{\iota}{S}$ can be computed by $\ORG{|S|} \circ \ANDG{k}\circ{\MODG{p}}$ circuit with  $O(l)$ wires.
\end{proof}

\section{Proof of \cref{nudet-exact}}
\begin{proof}
Let $\po a$ be an $n$-variate polynomial over $GF(p)$ with ABP of size $l$. We will create an $n$-ary polynomial $\po f_a$ of $\m A$ which will have the property that $\po f((x_1,y_1), \ldots, (x_n, y_n)) = (\po a(y_1, \ldots, y_n), 0)$ and the size of $\po f$ will be $O(\text{poly}(l))$. We do it by induction on the length of the process defining $\po a$. If $\po a$ is a constant 1, then a constant polynomial returning $(1,0)$ does the job. Now when we already have two ABPs $\po a,\po b$ which correspond to polynomials $\po f_a, \po f_b$ we create a polynomial corresponding to $\po a+\po b$ simply by adding them  with abelian addition $\po f_a+\po f_b$. When we have some ABP $\po b$ represented by $\po f_b$ and want to represent $y_i \cdot \po b$, we can achieve it by using a formula $\po v(\po f_b((x_1, y_1), \ldots, (x_n,y_n)), (x_i, y_i))$. We can use $\po v$ in the same way to represent multiplication by a constant. It follows that any ABP a of size $l$ can be simulated in this way by a polynomial/circuit $\po f_a$ of $\m A$ of size $O(l)$. In effect, $\nudet{p} \subseteq \palg{A}$, as for any polynomial $a$ with a short $ABP$ with accepting set $S$  we can instead use polynomial $\po f_{a}$ over $\m A$ with accepting set $S \times 0$ and $\iota(b) = (0,b)$ which accepts the same words.

For the inclusion $\palg{A} \subseteq \nudet{p}$, assume we have an $n$-ary circuit $\po f$ over $\m A$ of length $O(n^c)$. We will prove that $\po f$ is of the form $\po f((x_1, y_1), \ldots, (x_n, y_n)) = (a_{\po f}(x_1,x_2, \ldots, x_n, y_1, \ldots, y_n),  A_{\po f}(y_1, \ldots, y_n))$, where $a_f$ is $2n$-ary polynomial with $ABP_p$ of size $O(n^{c+1})$ and $A$ is an affine combination of variables $y_1, \ldots, y_n$.  Indeed, we proceed by induction on the structure of $\po f$. If $\po f$ is variable or a constant the statement is trivially satisfied. If $\po f$ is created by composing basic operations $+, \pi_1, \pi_2$ or $u$. with some shorter polynomials, the statement is also straightforward. So assume $\po f = \po v(\po f_1, \po f_2)$. By induction we have $ABP$s $a_{\po f_1}, a_{\po f_2}$ and affine functions $A_{\po f_1}, A_{\po f_2}$. We can see that $0$ is an affine function so the only thing we need to check is that we can represent $a_{\po f_1} \cdot A_{\po f_2}$ with a short ABP. But since $A_{\po f_2}$ is affine we can write it in a form $\alpha_1 y_1 + \ldots \alpha_n y_n + C$. Hence $a_{\po f_1} \cdot A_{\po f_2} = C\cdot a_{\po f_1} + \sum_{i=1}^n \alpha_i \cdot y_i \cdot a_{\po f_1}$. We can clearly see that now to create $a_{\po f_1} \cdot A_{\po f_2}$ we first multiply $a_{\po f_1}$ by variables $y_i$ then by constants $\alpha_i$ and $C$ and then sum the result which takes $O(n)$ steps (this explains why we need $O(n^{c+1})$ instead of just $O(n^c)$ size for $a_{\po f}$). 

Now  take arbitrary instruction $\iota $, an accepting set $S\subseteq A$ and an $n$-ary polynomial $\po f$ over $\m A$ of size $O(n^c)$ and we will create a polynomial-size $ABP_p$ $b_{\po f}$ with the property that $\po f(\iota(b_1), \ldots, \iota(b_n)) \in S$ iff $b_{\po f}(\overline b) = 1$. 
First note that $\iota(b) = (\iota_1(b), \iota_2(b))$ and $\iota_1, \iota_2$ can be written in a form $\iota_i(b) = \alpha_i b + \beta_i$ for $\alpha_i, \beta_i \in GF(p)$. Indeed every function from $\{0,1\}$ to $GF(p)$ can be interpolated by an affine function. 
Note that the statement $\po f(\iota(b_1), \ldots, \iota(b_n)) \in S$ is the disjunction of statement $\po f(\iota(b_1), \ldots, \iota(b_n)) = s$ for $s \in S$, which itself is a conjunction of two statements $a_{\po f}(\iota_1(b_1), \ldots, \iota_1(b_n), \iota_2(b_1), \ldots, \iota_2(b_n)) = s_1$ and $A_{\po f}(\iota_2(b_1), \ldots, \iota_2(b_n)) = s_2$ for $s = (s_1, s_2)$. As a result there exist a binary function $\po d$ from $GF(p)$ to $GF(p)$ with the image $\{0,1\}$ such that 
\[
\po d(a_{\po f}(\iota_1(b_1), \ldots, \iota_1(b_n), \iota_2(b_1), \ldots, \iota_2(b_n)), A_{\po f}(\iota_2(b_1), \ldots, \iota_2(b_n))) = 1
\]
iff $\po f(\iota(b_1), \ldots, \iota(b_n)) \in S$. Since both $a_{\po f}$ and $A_{\po f}$ are two ABPs of polynomial size from \cref{affine-subst} (applied to affine substitutions corresponding to $\iota_1, \iota_2$) and \cref{compose_poly_const} (applied to eliminate $d$) there exists an ABP $b_{\po f}$ which computes the desired function and is of polynomial size. 

So now, when we have a NuDFA given by $\iota$, accepting set $S \subseteq A$ and sequence of circuits $\po f_1, \po f_2, \ldots$ recognizing some language $L$, we can clearly recognize the same language with the sequence of $ABP$s $b_{\po f_1}, b_{\po f_2}, \ldots$ and accepting set $\{1\}$.
\end{proof}

\section{Proof of \cref{lem-nilp}}\label{apend-lem29}

We will need the following folklore lemma.

\begin{lemma}\label{lem:anddmodp}
Let $f$ be a function computed by a $\circuit{\BOUNDG{d'}}{\MODG{p}}{\ANDG{d}}$-circuit of length $l$. Then $f$ can be also computed by $\circuittwo{\MODG{p}}{\ANDG{d\cdot d'\cdot p}}$ circuit of size $O(\text{poly}(l))$.
\end{lemma}
\begin{proof}
Note that from every  $\MODG{p} \circ \ANDG{d}$ subcircuit we can create a polynomial over $GF(p)$ by replacing ${\ANDG{d}}$ by multiplication of corresponding inputs/variables and then by summing such created monomials. Let $\po w_1, \ldots, \po w_{d'}$ be all such created polynomials. Now notice that the value of our initial circuit $C$ evaluated on input $\ov b$ can be reconstructed from values $\po w_1(\overline{b}), \ldots \po w_{d'}(\overline{b})$. So there is a constant arity function $\po d$ such that $C(\overline{b}) = \po d(\po w_1(\overline{b}), \ldots, \po w_{d'}(\overline{b})))$. Note that $\po d$ can be represented by a polynomial over $GF(p)$ of degree $d'p$. So overall if we distribute multiplications in \\$\po d(\po w_1(\overline{b}), \ldots, \po w_{d'}(\overline{b})))$ we get a polynomial in a sparse form which is a sum of at most $l^{d'p}$ monomials of degree at most $dd'p$ over variables $\overline{b}$. So now if we replace multiplications by $\ANDG{dd'p}$ gate and addtion modulo $p$ by $\MODG{p}$ gate we get the desired circuit.
\end{proof}

\begin{proof}[Proof of \cref{lem-nilp}]
By \cite{fm} (Theorem 7.1) we know that $\m A$ has (up to isomorphism) the following form:
\begin{itemize}
    \item The universe $A$ is a direct product $M \times B$, where $M$ is an underlying set of a simple module $\m M$ and $B = A/\alpha$. For each $a \in A$ we write $a^{M}$ to the projection of $a$ to first coordinate and $a^{B}$ to the projection of $a$ to the second one. Similarly if $S \subseteq A$ we write $S^{M}$ or $S^{B}$ for the set of the projection of elements of $S$ to the respective coordinate.
    \item Every $k$-ary polynomial operation $\po p$ of $\m A$ has the form
    \begin{equation}\label{form-corr}
        \po p(\overline{a}) = (\po p^{\m M}(\overline{a}^M) +  \mathbf{\hat{p}}(\overline{a}^B), \po p^{\m B}(\overline{a}^B) )
    \end{equation}
    where $\po p^{\m M}$ is a polynomial of a module $\m M$, $\po p^{\m B}$ is a polynomial of an algebra $\m A/\alpha$ and $\hat{\po p}$ is a function from $B^k \mapsto M$ which we call a distortion function of $\po p$.
\end{itemize}

We can write even more detailed formula for how the polynomial $\po p$ can be expressed on its coordinates in terms of the basic operations its built of. Let $G$ be an output gate of $\po p$. Now by unwinding the formula \ref{form-corr} alongside the circuit definition of $\po p$ we can see that $\po p(\overline{a})^{M}$ is equal to
\[ \po p^{\po M}(\overline{a}^M) + \sum_{H \in U(G)} \beta_H \cdot \hat{\po f}_H((\po p^1_H(\overline{a}))^{M}, \ldots, (\po p^{\ar{f_H}}_H(\overline{a}))^{M}) 
\]
where $\beta_H$ are endomorphisms of $\po M$ and $\po p^1_H, \ldots, \po p^{\ar{f_H}}_H$ are subcircuits that are fed as inputs to a gate $H$, $U(G)$ refers to all inputs of the circuit $\po p$, and $\hat{\po f}_H$ is the basic operation of corresponding to gate $H$.

Note that to prove the Lemma it is enough to assume that set $S$ of the program $\progb{\po p}{\iota}{S}$ is a one element set. Indeed if $\progb{\po p}{\iota}{\{s\}}$ has the desired ${\MODG{p}}\circ{\ANDG{d}}\circ{\PROGG{\m A/\alpha}}$ circuit for $s \in S$, then by \cref{lem:anddmodp} also $\progb{\po p}{\iota}{S}$ (with a little bigger $d$), because alternative over $s \in S$ is a bounded arity function.

Now notice that $\progb{\po p}{\iota}{S}(\overline{b}) = 1$ iff its quotient program satisfies $\progb{\po p}{\iota/\alpha}{S/\alpha}(\overline{b}) = 1 $ and the value $\prog{ 
\po p}{\iota}(\overline{b})^{M} = \po p^{\m M}(\iota(\overline{b})^M) +  \mathbf{\hat{p}}(\iota(\overline{b})^B)$ belongs to the set $S^{M}$ (since $S$ is a one element set). Here $\iota/\alpha$ is just $\iota$ projected in the image to values in $\m A/\alpha$. Obviously the function 
$\progb{\po p}{\iota/\alpha}{S/\alpha}$ is computed (by definition) by $\PROGG{\m A/\alpha}$ circuit so if we prove that the Boolean evaluation of the statement  $\prog{ 
\po p}{\iota}(\overline{b})^{M} \in S^{M}$ can be computed by  $\circuitfour{\ANDG{d'}}{\MODG{p}}{\ANDG{d}}{\PROGG{\m A/\alpha}}$ circuit we will be done, because then we can express the conjunction of two conditions with $\circuitfour{\ANDG{d'+1}}{\MODG{p}}{\ANDG{d}}{\PROGG{\m A/\alpha}}$ circuit and use \cref{lem:anddmodp} to get read of $\ANDG{d'+1}$ layer. Thus now we focus on the expression $\po p^{\m M}(\overline{a}^M) +  \mathbf{\hat{p}}(\overline{a}^B)$.

Note that $M$ is an underlying set of a finite simple module. Hence it can be identified with an underlying set of an Abelian group $\mathbb{Z}_{p}^{\nu}$, where $p$ is a prime and $\nu$ is a natural number. Moreover every polynomial operation $\po p^M(x_1, \ldots, x_n)$ can be written in a form $\beta_1 x_1 + \ldots + \beta_n x_n + c$, where each $\beta_i$ is an endomorphism of an Abelian group $\mathbb{Z}_{p}^{\nu}$ and $c\in \mathbb{Z}_{p}^{\nu}$.

Now let $\pi_1, \ldots, \pi_\nu$ be projections to the respective coordinates of $\mathbb{Z}_{p}^\nu$. Note that the condition $\prog{ 
\po p}{\iota}(\overline{b})^{M} \in S^{M}$ is a conjunction of the projection conditions $\pi_i \prog{ 
\po p}{\iota}(\overline{b})^{M} \in \pi_i S^{M}$. So if we prove that the Boolean evaluation of the formula $\pi_i \prog{ 
\po p}{\iota}(\overline{b})^{M} \in \pi_i S^{M}$ can be computed by (polynomial size) ${\MODG{p}}\circ{\ANDG{d}}\circ{\PROGG{\m A/\alpha}}$ circuit it follows that $\prog{ 
\po p}{\iota}(\overline{b})^{M} \in S^{M}$ itself is computed by (polynomial size) $\circuitfour{\ANDG{\nu}}{\MODG{p}}{\ANDG{d}}{\PROGG{\m A/\alpha}}$ circuit so by \cref{lem:anddmodp}  also by $\circuit{\MODG{p}}{\ANDG{d}}{\PROGG{\m A/\alpha}}$ circuit (with possible bigger $d$). 

Now notice that the projection $\pi_i$ is linear map from the abelian group $(Z_p^{\nu}, +)$ to $Z_p$. What follows is that in the formula  \[
\pi_i\po p(\overline{a})^{M} = \pi_i\po p^{\po M}(\overline{a}^M) + \sum_{H \in U(G)} \pi_i\beta_H \cdot \hat{\po f}_H((\po p^1_H(\overline{a}))^{B}, \ldots, (\po p^{\ar{f_H}}_H(\overline{a}))^{B}) 
\]
the $\pi_i\po p^{\po M}(\overline{a}^M)$ term is just an affine $Z_p$-combination of the variables $\pi_i a_j^{M}$ (for all $1 \leq i \leq \nu$ and $1 \leq j \leq n$). It is because $\po p(\overline{a})^{M}$ is an operation of a module $M$ with the underlying abelian group being $Z_p^{\nu}$. Secondly the $\pi_i \beta_H$ is also a linear map from $(Z_p^{\nu}, +)$ to $Z_p$. So overall 
\[
\pi_i\po p(\overline{a})^{M} = \sum_{i=1}^l \gamma_i y_i
\]
where $\gamma_i$ are coefficients from $Z_p$ and each $y_i$ is either $\pi_i a_j^{M}$ or\\ $\pi_i \hat{\po f}_H((\po p^1_H(\overline{a}))^{B}, \ldots, (\po p^{\ar{f_H}}_H(\overline{a}))^{B})$.  So after applying the instruction $\iota$ we have

\[
\pi_i\prog{ 
\po p}{\iota}(\overline{b})^{M} = \sum_{i=1}^l \gamma_i y_i'
\]
where each  $y_i'$ is either $\pi_i \iota_j(b_j)^{M}$ or\\ $\pi_i \hat{\po f}_H(\po p^1_H(\iota(\overline{b}))^{B}, \ldots, (\po p^{\ar{f_H}}_H(\iota(\overline{b})))^{B})$

Now notice that in the former case the value 
\[
\pi_i \hat{\po f}_H(\po p^1_H(\iota(\overline{b}))^{B}, \ldots, (\po p^{\ar{f_H}}_H(\iota(\overline{b})))^{B})
\]
is fully determined by the value of the quotient programs $\progb{\po p^j_H}{\iota/\alpha}{c}$, for all $c\in B$, where $\po p^1_H$ is treated as a polynomial of $\m A/\alpha$. In other words there is a function $\po q^j_H(y_{1,1}, \ldots, y_{\ar{f_h}, |B|})$ of type\break $\{0,1\}^{|B|\cdot\ar{f_h}} \longrightarrow Z_p$ such that
\begin{align*}
\pi_i \hat{\po f}_H(&\po p^1_H(\iota(\overline{b}))^{B}, \ldots, (\po p^{\ar{f_H}}_H(\iota(\overline{b})))^{B}) = \\ = &\po q^j_H(\progb{\po p^1_H}{\iota/\alpha}{c_1}(\overline{b}), \ldots, \progb{\po p^{\ar{f_h}}_H}{\iota/\alpha}{c_{|B|}}(\overline{b}))
\end{align*}
where $c_1, \ldots c_{|H|}$ traverse all values of $B$. Now $\po q^j_H$ has arity $\ar{f_h} \cdot |B|$ and operates on Boolean values $\{0,1\}$ so it can be represented by a bounded degree polynomial over $Z_p$ of degree $\ar{f_h} \cdot |B|$. Moreover the function given by the instruction $\iota_j(b_j)$  maps Boolean values $\{0,1\}$ into $Z_p$, so it can be exactly interpolated by an affine transformation. So $\pi_i\prog{ 
\po p}{\iota}(\overline{b})^{M}$ can be expressed as a polynomial of degree $\ar{f_h} \cdot |B|$ which is evaluated on both variables $b_j$ and quotient programs $\progb{\po p^j_H}{\iota/\alpha}{c}$. Moreover the condition $\pi_i\prog{ 
\po p}{\iota}(\overline{b})^{M} \in S_m = \{s_m\}$ is satisfied iff.\! 
\[
0 = \Big(\pi_i\prog{ 
\po p}{\iota}(\overline{b})^{M} - s_m\Big).
\]
Hence the condition $\pi_i\prog{ 
\po p}{\iota}(\overline{b})^{M} \in S_m$  might be represented with a polynomial of degree $\ar{f_h} \cdot |B| \leq |A| \cdot \ar{f_h}$ evaluated on variables $b_j$ and programs $\progb{\po p^j_H}{\iota/\alpha}{c}$. Now we can see that
\begin{enumerate}
    \item $\progb{\po p^j_H}{\iota/\alpha}{c}$ is in $\PROGG{A/\alpha}$
    \item multiplication in every monomial can be represented with $\ANDG{|A| \cdot \ar{f_h}}$ gate
    \item addition in the polynomial can be represented with one $\MODG{p}$ gate.
\end{enumerate}
So overall the condition $\pi_i\prog{ 
\po p}{\iota}(\overline{b})^{(M)} \in S_m$ can be evaluated by $\circuit{\MODG{p}}{\ANDG{|A| \cdot \ar{f_h}}}{\PROGG{A/\alpha}}$ circuit. Since $\ar{f_H}$ is arity of a basic operation of $\m A$, then it is just a constant dependent on $\m A$. This finishes the proof.
\end{proof}

\section{Proof of \cref{ex-type-5}}

Note that since $\po e_1$, $\po e_2$, $\po e_3$ are basic operations of algebras $\m A_1$, $\m A_2$, $\m A_3$, $\m A_4$ these algebras are simple. Moreover, the only mappings $\iota$ for which NuDFA's over $\m A_i$'s are able to recognize interesting language are this for which $\bot\not\in\iota(\set{0,1})$. It is because circuits over above algebras has values $\bot$ whenever even one of their input gate has this value. $\palg{\m A_1|_{\set{0,1}}}\subseteq\palg{\m A_i}$ give us that  $\MON\subseteq\palg{\m A_2}$ and $\palg{\m A_3}=\ppoly$. On the other hand basic operations of $\m A_1$ and $\m A_2$ preserve total order $\bot<0<1$ and hence $\palg{\m A_1}$ and $\palg{\m A_2}$ are contained in $\MMON$. To see that $\MON\nsubseteq\palg{\m A_1}$, note that $\po e_i(x\wedge y)=\po e_i(x)\wedge\po e_i(y)$ and $\wedge$ is just a minim operation on the set $\bot<0<1$. Therefore, every circuit over $\m A_1$ can be expressed as a minimum of unary  monotone functions applied to the input gates in which value of each input gate is used at most once. Now, it is not hard to see that language recognized by the sequence of monotone CNF formulas $c_n=\bigwedge_{i=1,\ldots, n}(x_{2\cdot i}\vee x_{2\cdot i +1})$ is not in $\palg{\m A_1}$. Finally, observe that $\po e_2(x)=x+x+1$, $\po e_3(x)=x+x$, $\po e_1(1)=0$ and $\po e_1(x_1+x_2+\ldots+x3)\not=\bot$ iff $x_1+x_2+\ldots+x3=1$. These facts give us that $\po e_1(x+\po e_1(y))=\po e_1(x)+\po e_1(y)$ and hence that every circuit over $\m A_4$ can be expressed in the form $\sum_{i\in I}x_i+\sum_{j\in J}f_j(\sum_{i\in I_j}x_i)$. Now it is obvious that  $P_{\m A_4}\subseteq \mathsf{AC}_k[2]$ for some small constant $k$.
\end{document}